\documentclass[twocolumn,conference,letterpaper]{IEEEtran}
\usepackage[left=0.75in,right=0.75in,top=1in,bottom=0.75in]{geometry}
\usepackage{amsmath}          
\usepackage{amssymb}          
\usepackage{amsthm}           
\usepackage{enumitem}         
\usepackage{bm}               

\newtheorem{assumption}{Assumption}
\newtheorem{remark}{Remark}
\newtheorem{problem}{Problem}
\newtheorem{lemma}{Lemma}
\newtheorem{theorem}{Theorem}

\theoremstyle{definition}
\newtheorem{definition}{Definition}
\usepackage{color}
\newcommand{\Ujoint}{\mathcal{U}_{\mathrm{joint}}}
\newcommand{\Ubox}{\mathcal{B}}
\newcommand{\bu}{\mathbf{u}}
\newcommand{\buc}{\mathbf{u}_c}
\newcommand{\bG}{\mathbf{G}}
\newcommand{\bn}{\mathbf{n}}
\newcommand{\bP}{\mathbf{P}}
\newcommand{\bI}{\mathbf{I}}
\newcommand{\bg}{\mathbf{g}}
\newcommand{\bq}{\mathbf{q}}

\newcommand{\diag}{\operatorname{diag}}
\newcommand{\Sbox}{S_{\mathrm{box}}}

\usepackage{color}
\usepackage{graphicx}
\usepackage{float,subcaption}
\newcommand{\tb}[1]{\textcolor{blue}{#1}}

\usepackage{hyperref}
\hypersetup{colorlinks, breaklinks, citecolor=blue, linkcolor=blue, urlcolor=blue}
\usepackage[nameinlink]{cleveref}
\Crefformat{figure}{#2Fig.~#1#3}
\Crefmultiformat{figure}{Figs.~#2#1#3}{ and~#2#1#3}{, #2#1#3}{ and~#2#1#3}
\Crefname{assumption}{Assumption}{Assumptions}
\usepackage{tikz}
\usetikzlibrary{arrows.meta,calc,positioning}
\IEEEoverridecommandlockouts

\begin{document}

\title{Admissibility-Preserving Control for Multi-Input Systems with Joint Capacity Constraints}

\author{Saurabh Kumar, Lohitvel Gopikannan, Shashi Ranjan Kumar, and Abhinav Sinha 
\thanks{S.Kumar, L. Gopikannan and S. R. Kumar are with the Intelligent Systems \& Control (ISaC) Lab, Department of Aerospace Engineering, Indian Institute of Technology Bombay, Mumbai 400076, India. (e-mails: saurabh.k@aero.iitb.ac.in, lohitvel@aero.iitb.ac.in, srk@aero.iitb.ac.in). A. Sinha is with the Guidance, Autonomy, Learning, and Control for Intelligent Systems (GALACxIS) Lab, Department of Aerospace Engineering, University of Cincinnati, OH 45221, USA. (e-mail: abhinav.sinha@uc.edu)}
}

\maketitle
	\thispagestyle{empty}
\begin{abstract}
This paper addresses the control of multi-input strict-feedback nonlinear systems subject to a \emph{joint capacity constraint}, in which the admissible input set is a coupled subset of the individual actuator limits. Unlike existing constraint-handling methods that enforce actuator bounds channel by channel and may unnecessarily suppress admissible control directions, we develop an Anisotropic Joint-Admissibility-Preserving Input Realization (AJ-APIR) framework that explicitly exploits the geometry of the joint constraint. The proposed realization constructs a state-dependent gain matrix whose spectral decomposition separates the commanded input into normal and tangential directions relative to the constraint boundary. The normal component is attenuated as the boundary is approached, while the tangential component is preserved, which allows the admissible control effort to be redistributed without loss of tracking authority. Integrated with a backstepping controller, the AJ-APIR framework guarantees forward invariance of the joint admissible set for all time. We establish exponential convergence of the tracking error to zero together with uniform boundedness of all closed-loop signals, and characterize the resulting command-demand behavior under the joint constraint. Simulation results for a representative second-order, two-input nonlinear system subject to a power-budget constraint demonstrate the efficacy of the proposed method to enforce the joint input constraint.
\end{abstract}
\section{Introduction}
In many practical multi-actuator systems, actuator channels share a common physical resource, such as the electrical bus shared by control surfaces and propulsion in fixed-wing UAVs, the thermal path shared by joint motors in redundantly actuated manipulators, or the prime mover shared by azimuth thrusters in marine surface vessels. This physical coupling gives rise to a joint constraint on the collective input that independent per-channel bounds fail to capture. In practice, the aggregate demand may breach the shared operating limit even when every individual channel remains within its own admissible range. Consequently, to maintain the collective input within the jointly admissible set at all times while ensuring that the system output tracks a prescribed reference trajectory remains an open challenge concerning safety-critical control.

The problem of designing feedback controllers for nonlinear systems subject to actuator constraints has attracted sustained research attention for more than three decades.  Nevertheless, the systematic treatment of coupled actuator constraints, particularly without unnecessarily sacrificing admissible control authority, remains an open problem. A foundational body of work addresses this problem through anti-windup augmentation of pre-designed linear controllers, in which a secondary compensator is introduced to recover closed-loop performance when the primary loop saturates \cite{teel1999anti,zaccarian2002common,grimm2003antiwindup}. Although systematic and amenable to convex synthesis, anti-windup methods are generally developed for linear nominal plants and do not readily accommodate the state-dependent gain structure of nonlinear strict-feedback systems. A parallel stream exploits the recursive structure of strict-feedback dynamics through backstepping \cite{krstic1995nonlinear}, with constraint handling achieved by confining each virtual control to a bounded feasible set. The authors in \cite{wen2011robust} showed that augmenting a standard backstepping design with auxiliary dynamics whose equilibrium tracks the saturated channel yields a complete Lyapunov analysis for tracking under asymmetric saturation, and the construction was subsequently extended to uncertain multi-input systems in \cite{chen2011adaptive}. It is worth noting that these treatments bound each channel separately, so that the feasible set they enforce is, by construction, a Cartesian product.

Barrier Lyapunov functions (BLFs) provide an alternative mechanism for enforcing constraints in strict-feedback systems by introducing a potential that grows without bound as a constrained variable approaches its prescribed limit, thereby preventing the trajectory from reaching the boundary. The authors in \cite{tee2009barrier} established asymptotic tracking for single-output systems subject to output constraints using BLFs. In \cite{ngo2005integrator}, this framework was discussed for full-state constraints, while \cite{liu2016barrier} combined BLFs with adaptive neural networks to address uncertain strict-feedback systems. However, standard BLF constructions assign an independent barrier potential to each constrained channel, resulting in a Lyapunov function composed of decoupled terms. Consequently, they do not directly capture the inter-channel coupling induced by a joint capacity constraint.

Control Barrier Functions (CBFs) \cite{ames2016control} provide a unified framework for barrier-certificate and Lyapunov-based stability analysis, with a quadratic program (QP)-based safety filter applicable to a broad class of nonlinear systems. The safety-critical control framework in \cite{ames2019control} accommodates general smooth constraint sets, including joint constraints, by incorporating barrier-function gradients into a pointwise QP. The authors in \cite{rauscher2016constrained} demonstrated the compatibility of such constraint-handling mechanisms with backstepping for relative-degree-one input constraints. Nevertheless, the QP formulation modifies the nominal control according to the Euclidean minimum-norm criterion, which is effectively direction-blind with respect to the geometry of the joint constraint. Near a joint constraint boundary, this may unnecessarily attenuate components of the commanded input that are tangent to the boundary and therefore do not contribute to constraint violation. Thus, constraint enforcement need not be isotropic, yet this geometric structure does not appear to be explicitly exploited in the existing CBF-based formulations.

A further class of methods relies on control allocation to map a scalar or low-dimensional virtual control into a higher-dimensional actuator command subject to box and joint constraints, typically by solving a weighted least-squares problem or quadratic program at each instant \cite{durham1993constrained,bodson2002evaluation,petersen2006constrained,johansen2013control}. The relationship between control allocation and optimal control was discussed in \cite{harkegaard2005resolving}, while a dynamic counterpart was developed in \cite{zaccarian2009dynamic}. Control allocation decouples the outer-loop control design from the constraint-satisfaction problem, making it attractive for redundant platforms. However, the feasibility of the allocation problem depends on the actuator authority available at each instant and is therefore not guaranteed during large-signal transients. Moreover, the stability of the interconnected outer-loop and allocator dynamics generally requires a separate analysis. Governor and predictive architectures \cite{garone2017reference,rawlings2020model}, in contrast, address feasibility at the planning or reference-generation level by reshaping the reference or optimizing over a finite horizon, respectively, but at a potentially high computational cost for fast actuation loops.

The Admissibility-Preserving Input Realization (APIR) framework, introduced for strict-feedback systems with asymmetric single-channel actuator constraints \cite{kumar2026admissibility}, adopts a dynamical-systems perspective on constraint satisfaction. Rather than modifying a nominal controller through a pointwise optimization or treating constraint enforcement as a separate allocation layer, APIR embeds admissibility directly into the physical input dynamics through a differential realization whose gain vanishes at the constraint boundary. A distributed extension to networked systems with directed communication graphs was subsequently developed in \cite{sinha2026networked}. These works establish an intrinsic alternative to optimization-based safety filters \cite{ames2016control,ames2019control} and allocation-based constraint handling \cite{johansen2013control}, in which admissibility is enforced through the dynamics of the realized input itself.

However, the existing APIR construction \cite{kumar2026admissibility} is fundamentally channel-wise, that is, its realization is formed from independent scalar admissibility factors and therefore applies to admissible sets that are Cartesian products of individual actuator ranges. This indicates that it may not directly address a coupled input envelope whose boundary depends on the collective actuator state. This may appear as a limitation, which is nonetheless shared by standard channel-wise saturation and backstepping- based controllers \cite{wen2011robust,chen2011adaptive} and by BLF formulations that impose separate barrier terms on individual constrained variables \cite{tee2009barrier,ngo2005integrator,liu2016barrier}. CBF-based formulations can represent general smooth joint constraint sets \cite{ames2016control,ames2019control}. However, their standard optimization-based realization modifies the nominal command through a pointwise safety filter rather than embedding the joint admissibility mechanism directly into the input dynamics. Likewise, control-allocation methods explicitly account for coupled actuator limitations, but do so by introducing a separate allocation layer between the outer-loop controller and the physical actuators \cite{johansen2013control}. It must be noted that this scheme may be effective for distributing control effort among constrained actuators, but the architecture separates constraint satisfaction from the closed-loop analysis. However, the feasibility must be maintained by the allocation problem, and the effect of the allocation on the stability and tracking properties of the overall closed-loop system warrants a separate analysis. To the best of our knowledge, the existing literature does not provide a dynamic realization that simultaneously enforces a non-Cartesian joint input envelope and preserves admissible control authority in directions tangent to its boundary.

The current work extends the notion of the admissibility-preserving control framework \cite{kumar2026admissibility} from individual actuator constraints to coupled multi-input constraints. We replace the channel-wise scalar realization with a matrix-valued, geometry-aware realization that differentiates between control directions that may risk joint admissibility and those that remain feasible. In particular, the proposed AJ-APIR selectively suppresses the component normal to the joint constraint boundary as it is approached, while preserving the tangential component. Thus, the proposed realization enforces the shared input envelope inherently while avoiding the unnecessary attenuation of admissible control redistribution that arises from direction-independent constraint handling. Our main contributions are summarized as follows:
\begin{itemize}
\item We propose an anisotropic joint-admissibility-preserving input realization (AJ-APIR) for multi-input strict-feedback nonlinear systems subject to coupled joint capacity constraints. This is unlike the channel-wise methods \cite{wen2011robust,chen2011adaptive}, which cannot express a coupled set.
\item The proposed realization employs a matrix-valued gain that splits spectrally into normal and tangential components relative to the joint constraint boundary. This allows the normal component of the commanded input to be attenuated near the boundary while retaining the tangential component. Thus, the proposed strategy extends the admissibility-preserving control of \cite{kumar2026admissibility} from channel-wise constraints to a non-Cartesian joint admissible set.
\item We integrate AJ-APIR with recursive backstepping and establish the corresponding closed-loop guarantees. In particular, we derive a compact forward-invariant subset of the joint admissible set on which the realization channel remains uniformly nondegenerate.
\item An AJ-APC compatibility condition is then introduced to relate the available actuator authority to the tracking demand, under which the tracking and input realization errors converge exponentially, and all closed-loop signals remain uniformly bounded.
\end{itemize}
We demonstrate the effect of exploiting the joint-constraint geometry through two numerical studies and an aerospace application. The first study shows that satisfying the individual actuator bounds does not ensure joint admissibility: the channel-wise isotropic realization repeatedly crosses the shared capacity boundary, whereas AJ-APIR remains inside the joint admissible set for the same commanded input. The second study considers a nonlinear two-input strict-feedback system and shows that the proposed realization maintains both the individual and joint constraints while the output converges to its desired reference trajectory. We further implement the proposed method for UAV constrained path following in 3D, which further bolsters its effectiveness.

\section{Problem Formulation and AJ-APIR}
\label{sec:problem}
Consider a multi-input strict-feedback nonlinear system of the form
\begin{subequations}
\label{eq:system}
\begin{align}
  \dot{x}_i &= f_i(\bar{x}_i) + g_i(\bar{x}_i)\,x_{i+1},\quad i = 1,\ldots,n-1, \label{eq:system_a}\\
  \dot{x}_n &= f_n(\bar{x}_n) + \mathbf{g}_n^{\top}(\bar{x}_n)\,\bu,\label{eq:system_b}\\
  y         &= x_1, \label{eq:system_c}
\end{align}
\end{subequations}
where $\bar{x}_i \triangleq [x_1,\ldots,x_i]^{\top} \in \mathbb{R}^i$ denotes the partial state vector, $\bu = [u_1,\ldots,u_m]^{\top} \in \mathbb{R}^m$ is the  control input vector, $\mathbf{g}_n : \mathbb{R}^n \to \mathbb{R}^m$ is the input coupling vector at the final integrator ($n\textsuperscript{th}$) stage, and $y \in \mathbb{R}$ is the measurable system output. The functions $f_i : \mathbb{R}^i \to \mathbb{R}$ and $g_i : \mathbb{R}^i \to \mathbb{R}$, $i = 1,\ldots,n-1$, represent the drift and coupling dynamics of each intermediate stage, respectively. Note that the single-output, multi-input structure in \eqref{eq:system} is representative of a broad class of aerospace and robotic platforms, wherein a single performance channel is regulated through coordinated action of multiple actuators.
\begin{definition}[Admissible input sets]
Each input channel $u_k$, $k = 1,\ldots,m$, is independently subject to asymmetric bounds
\begin{equation}
  \label{eq:box_constraints}
  u_k \in \mathcal{U}_k \triangleq
  \bigl(u_{\min}^{(k)},\; u_{\max}^{(k)}\bigr),
  \quad u_{\min}^{(k)} < 0 < u_{\max}^{(k)},
\end{equation}
so that the Cartesian product $\Ubox \triangleq \prod_{k=1}^{m}\mathcal{U}_k$ defines the \emph{individual admissible set}. In addition to these per-channel limits, the $m$ actuators share a common operating envelope defined by the \emph{joint capacity constraint} 
\begin{equation}
  \label{eq:joint_constraint}
  \varphi(\bu) < 0,
\end{equation}
where $\varphi : \mathbb{R}^m \to \mathbb{R}$ is a smooth function whose zero level set $\{\bu : \varphi(\bu) = 0\}$ defines the boundary of the shared capacity envelope. The resulting \emph{joint admissible set} is
\begin{equation}
  \label{eq:joint_set}
  \Ujoint \triangleq
  \Ubox \cap
  \bigl\{\bu \in \mathbb{R}^m : \varphi(\bu) < 0\bigr\}.
\end{equation}
\end{definition}
The control objective is to synthesize $\bu(t)$ such that $\bu(t)\in\Ujoint$ for all $t\geq0$ while driving the output $y(t)$ to the prescribed reference trajectory.

We formulate the multi-input APIR by applying the scalar saturation mechanism independently to each input channel. We refer to this formulation as \emph{isotropic APIR}.
\begin{equation}\label{eq:isotropic_APIR}
  \dot{\bu} = \mathbf{D}(\bu)\buc - p_2\bu,
\end{equation}
where $\buc \in \mathbb{R}^m$ is the commanded input vector, $p_2 > 0$ is a damping parameter, and
\begin{equation} \label{eq:G_iso}
\mathbf{D}(\bu)
\triangleq
\diag\left( p_1^{(1)} S^{(1)}(u_1),\ldots, p_1^{(m)} S^{(m)}(u_m)\right)
\end{equation}
is a diagonal, channel-wise gain matrix. Here, $S^{(k)}(u_k)$ is the channel-$k$ admissibility factor, defined as
\begin{equation}\label{eq:Sk}
S^{(k)} \triangleq
\rho^{(k)}\left[1 \!-\! \left(\tfrac{u_k}{u_{\max}^{(k)}}\right)^{\gamma^{(k)}}\right]
\!+\! \left(1\!-\!\rho^{(k)}\right)\left[1 \!-\! \left(\tfrac{u_k}{u_{\min}^{(k)}}\right)^{\gamma^{(k)}}\right],
\end{equation}
where $\rho^{(k)} = 1$ if $u_k > 0$ and $\rho^{(k)} = 0$ otherwise, and $\gamma^{(k)}$ is a positive even integer.

Although \eqref{eq:isotropic_APIR} guarantees forward invariance \cite{kumar2026admissibility} of $u_k \in \mathcal{U}_k$, it attenuates $\buc$ uniformly across all directions of the $m$-dimensional input space. When $\bu$ is near the joint constraint boundary $\{\bu : \varphi(\bu) = 0\}$, the gradient of $\varphi$ defines a distinguished direction, the outward unit normal
\begin{equation}
  \label{eq:normal}
  \bn(\bu) \triangleq
  \frac{\nabla\varphi(\bu)}{\|\nabla\varphi(\bu)\|},
\end{equation}
which is the direction along which the capacity constraint can be violated.

Commands directed along $\bn(\bu)$ may drive the state toward the boundary, which may risk constraint violation, and therefore require strong attenuation. In contrast, commands lying in the $(m{-}1)$-dimensional tangent hyperplane $\mathcal{T}(\bu) \triangleq \left\{\mathbf{d} \in \mathbb{R}^m : \bn^{\top}(\bu),\mathbf{d} = 0\right\}$ do not change $\varphi(\bu)$. Hence, they require no such attenuation. By treating all directions symmetrically, the isotropic realization \eqref{eq:isotropic_APIR} attenuates tangential commands, increasing the commanded input demand and degrading tracking performance. This loss of directional selectivity motivates the anisotropic formulation developed in the sequel.

To solve the limitations discussed above, we propose replacing the diagonal, isotropic gain in \eqref{eq:G_iso} with a direction-aware matrix-valued gain that distinguishes normal from tangential directions relative to the joint constraint boundary. Let us define the tangential projector
\begin{equation}\label{eq:projector}
  \bP(\bu) \triangleq \bI_m - \bn(\bu)\,\bn^{\top}(\bu),
\end{equation}
where $\bI_m$ is the $m \times m$ identity matrix. 
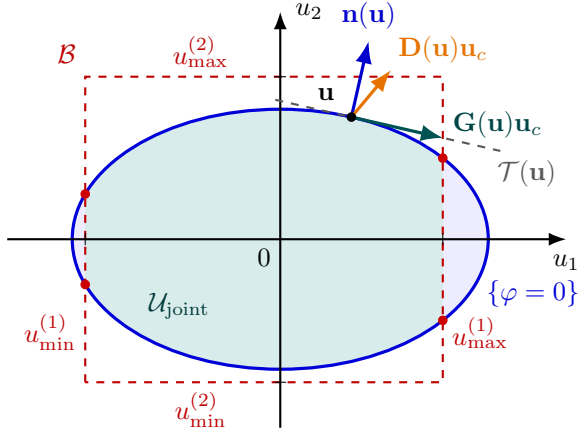
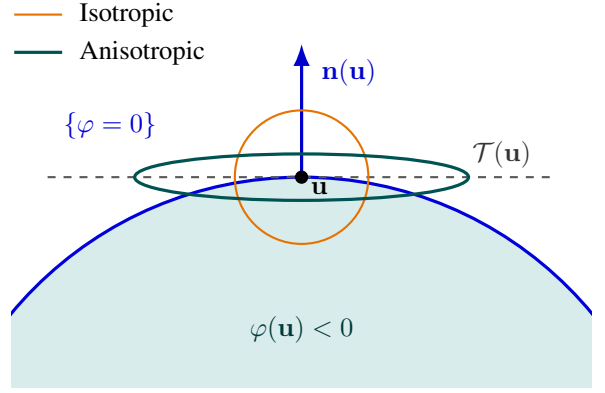
\begin{figure*}[!ht]
\centering
\begin{subfigure}[c]{0.52\linewidth}
\centering
\begin{tikzpicture}[scale=0.86,>=Latex]
    \def\uminone{-3.0}
    \def\umaxone{2.5}
    \def\umintwo{-2.2}
    \def\umaxtwo{2.5}
    \fill[blue!7] (0,0) ellipse (3.2 and 2.0);
    \begin{scope}
    \clip (\uminone,\umintwo) rectangle (\umaxone,\umaxtwo);
    \fill[teal!15] (0,0) ellipse (3.2 and 2.0);
    \end{scope}
    \draw[red!75!black, thick, dashed] (\uminone,\umintwo) rectangle (\umaxone,\umaxtwo);
    \draw[blue!85!black, very thick] (0,0) ellipse (3.2 and 2.0);
    \foreach \c in {(\umaxone,1.248),(\umaxone,-1.248),(\uminone,0.696),(\uminone,-0.696)}
    {\fill[red!80!black] \c circle (2.1pt);
    }
    \node[red!75!black,  left] at (\uminone,2.85) {$\Ubox$};
    \node[blue!85!black,  right=1em, above] at (3.48,-1.18) {$\{\varphi=0\}$};
    \node[teal!45!black] at (-1.55,-1.02) {$\Ujoint$};
    \draw[->, thick] (-4.2,0) -- (4.4,0) node[below=2pt] {$u_1$};
    \draw[->, thick] (0,-2.9) -- (0,3.5) node[right=2pt] {$u_2$};
    \node[below left] at (0,0) {$0$};
    \draw (\umaxone,-0.09) -- (\umaxone,0.09);
    \node[red!75!black,below =3.5em,right] at (\umaxone,0) {$u_{\max}^{(1)}$};
    \draw (\uminone,-0.09) -- (\uminone,0.09);
    \node[red!75!black,below =3.5em,left] at (\uminone,0) {$u_{\min}^{(1)}$};
    \draw (-0.09,\umaxtwo) -- (0.09,\umaxtwo);
    \node[red!75!black,left = 3em, above] at (0,\umaxtwo) {$u_{\max}^{(2)}$};
    \draw (-0.09,\umintwo) -- (0.09,\umintwo);
    \node[red!75!black,left=3em, below] at (0,\umintwo) {$u_{\min}^{(2)}$};
    \coordinate (U) at (1.094,1.879);
    \draw[thick, dashed, gray!70!black] ($(U)-(-12.8:1.2)$) -- ($(U)+(-12.8:2.35)$);
    \node[gray!70!black,  right=1em,below] at ($(U)+(-12.8:2.35)$) {$\mathcal{T}(\bu)$};
    \draw[->, very thick, blue!80!black] (U) -- ++(77.2:1.20);
    \node[blue!80!black,  right=4pt, above=1pt] at ($(U)+(77.2:1.20)$) {$\bn(\bu)$};
    \draw[->, very thick, orange!90!black] (U) -- ++(50:0.95);
    \node[orange!90!black,  above right=-0.1em] at ($(U)+(50:0.95)$) {$\mathbf{D}(\bu)\buc$};
    \draw[->, very thick, teal!65!black] (U) -- ++(-12.8:1.45);
    \node[teal!55!black,  above=0.5em, right] at ($(U)+(-12.8:1.45)$) {$\bG(\bu)\buc$};
    \fill[black] (U) circle (2.1pt);
    \node[above left=4pt] at (U) {$\bu$};
\end{tikzpicture}
\caption{Admissible input geometry and normal--tangent decomposition of the commanded input at $\bu$.}
\label{fig:sets_a}
\end{subfigure}
\hfill
\begin{subfigure}[c]{0.45\linewidth}
\centering
\begin{tikzpicture}[scale=1.92,>=Latex]
    \def\Rb{2.6}
    \begin{scope}
    \clip (-2.0,-1.45) rectangle (2.0,1.25);
    \fill[teal!15] (0,-\Rb) circle (\Rb);
    \draw[blue!85!black, very thick] (0,-\Rb) circle (\Rb);
    \end{scope}
    \coordinate (V) at (0,0);
    \draw[gray!70!black, thick, dashed] (-1.75,0) -- (1.75,0);
    \node[gray!40!black,  above] at (1.38,0) {$\mathcal{T}(\bu)$};
    \draw[->, blue!80!black, very thick] (V) -- (0,0.92);
    \node[blue!80!black,  right=4pt] at (0,0.72) {$\bn(\bu)$};
    \draw[orange!90!black, thick] (V) circle (0.46);
    \draw[teal!65!black, very thick] (V) ellipse (1.15 and 0.16);
    \fill (V) circle (1.3pt);
    \node[ below=4pt,right] at (V) {$\bu$};
    \node[ teal!45!black] at (0,-1.05) {$\varphi(\bu)<0$};
    \node[ blue!85!black] at (-1.32,0.36) {$\{\varphi=0\}$};
    \draw[orange!90!black, thick] (-1.98,1.12) -- (-1.65,1.12);
    \node[ right=3pt] at (-1.65,1.12) {Isotropic};
    \draw[teal!65!black, very thick] (-1.98,0.86) -- (-1.65,0.86);
    \node[ right=3pt] at (-1.65,0.86) {Anisotropic};
\end{tikzpicture}
\caption{Realizable-velocity sets at a boundary point for isotropic and anisotropic input realization.}
\label{fig:sets_b}
\end{subfigure}
\caption{Geometric interpretation of the anisotropic admissibility-preserving input realization.}
\label{fig:sets}
\end{figure*}
\begin{definition}[AJ-APIR]
The anisotropic joint-admissibility-preserving input realization is defined by the differential equation
\begin{equation}\label{eq:AJ_APIR}
\dot{\bu} = \bG(\bu)\,\buc - p_2\,\bu,
\end{equation}
with the anisotropic gain matrix
\begin{equation}\label{eq:G_aniso}
  \bG(\bu) \triangleq G_{\perp}(\bu)\,\bn(\bu)\,\bn^{\top}(\bu) + G_{\parallel}(\bu)\,\bP(\bu),
\end{equation}
where $G_{\perp}(\bu) > 0$ and $G_{\parallel}(\bu) > 0$ are scalar gain functions that modulate the realization speed in the normal and tangential directions, respectively.
\end{definition}
Note that $G_{\perp}(\bu)$ is required to vanish as $\bu$ approaches the joint constraint boundary $\{\varphi(\bu) = 0\}$ from the interior of $\Ujoint$, thereby preventing the input from escaping through $\varphi(\bu) = 0$, while $G_{\parallel}(\bu)$ is permitted to remain bounded away from zero so that tangential realization is unimpeded near the boundary.

\Cref{fig:sets} illustrates the above mentioned effects. In \Cref{fig:sets_a}, the joint admissible set $\Ujoint$ is the intersection of the box $\Ubox$ with the shared envelope $\{\varphi<0\}$, and the marked points are the corners at which a box face meets the envelope. At the boundary point $\bu$, the joint constraint is active while both box constraints are slack, so every channel factor $S^{(k)}(u_k)$ in \eqref{eq:Sk} is strictly positive there. The isotropic gain \eqref{eq:G_iso} vanishes only as a channel approaches its own bound and is therefore unattenuated at such a point, so $\mathbf{D}(\bu)\buc$ has a component along $\bn(\bu)$ and nothing in the realization opposes motion across the envelope. The anisotropic gain \eqref{eq:G_aniso} instead annihilates the normal direction through $G_\perp$, so that $\bn^{\top}\bG(\bu)\buc = G_\perp\,\bn^{\top}\buc = 0$ and the commanded contribution lies in $\mathcal{T}(\bu)$. \Cref{fig:sets_b} shows the same contrast through the sets of velocities realizable at $\bu$. A sphere for the isotropic gain, and for \eqref{eq:G_aniso} an ellipsoid flattened along $\bn(\bu)$ but unattenuated on $\mathcal{T}(\bu)$.

\begin{remark} \label{rem:spectral}
The gain matrix \eqref{eq:G_aniso} is symmetric positive definite with eigenvalue $G_{\perp}(\bu)$ in the direction of $\bn(\bu)$ and eigenvalue $G_{\parallel}(\bu)$ in the $(m{-}1)$-dimensional orthogonal complement. When $G_{\perp} \ll G_{\parallel}$ near $\{\varphi = 0\}$, the ellipsoid of realizable velocities is highly elongated along the boundary tangent plane, meaning that the realization dynamics can track boundary-parallel commands at full speed while simultaneously preventing outward boundary crossing. In contrast, the isotropic gain \eqref{eq:G_iso} corresponds to a spherical realizable-velocity ellipsoid that offers no directional discrimination. It is worth noting that \eqref{eq:G_aniso} reduces to the isotropic diagonal form \eqref{eq:G_iso} when $\varphi$ is absent and $G_\perp = G_\parallel = \diag(p_1^{(k)}S^{(k)})$, which shows that AJ-APIR generalizes the existing single-input framework to a multi-input setting.
\end{remark}
The following standard assumptions are made.
\begin{assumption} \label{asm:reference}
The reference output trajectory $y_d(t)$ and its derivatives up to order $n+1$ are uniformly bounded and piecewise continuously differentiable.
\end{assumption}
\begin{assumption} \label{asm:model}
The functions $f_i$, $g_i$ $(i{=}1,\ldots,n)$ are continuously differentiable. Moreover, $g_i(\bar{x}_i) \neq 0$ for all $\bar{x}_i$ in the region of interest, and the input coupling vector satisfies $\sigma_{\min}(\mathbf{g}_n^{\top}(\bar{x}_n)) \geq g_{\min} > 0$, where $\sigma_{\min}(\cdot)$ denotes the minimum singular value.
\end{assumption}

The non-vanishing condition on $g_i$ excludes loss of controllability at intermediate integrator stages, which would otherwise render the virtual control design ill-conditioned. The lower bound $g_{\min}$ in \Cref{asm:model} ensures that the final-stage coupling matrix $\bg_n^{\top}$ remains persistently non-degenerate, so that every commanded input $u_c^*$ is achievable by some physical $\bu$.

\begin{assumption} \label{asm:phi}
The joint constraint function $\varphi : \mathbb{R}^m \to \mathbb{R}$ is twice continuously differentiable and convex. Its gradient satisfies $\|\nabla\varphi(\bu)\| \geq \varphi_{\min} > 0$	on a neighborhood of $\{\varphi = 0\}$, and $\Ujoint$ is nonempty, connected, and satisfies $\mathbf{0} \in \mathrm{int}(\Ujoint)$. 
\end{assumption}

The convexity and origin-inclusion conditions in \Cref{asm:phi} imply $\nabla\varphi^{\top}(\bu)\,\bu \geq 0$ for all $\bu$ on the boundary $\{\varphi = 0\}$, by the supporting hyperplane theorem \cite{boyd2004convex_chapter}. Note that when the normal gain $G_\perp$ vanishes at the joint boundary, the residual dynamics point strictly inward. For physical actuators, the origin, which represents zero net actuation, is invariably a feasible operating point, so the origin-inclusion condition is mild.
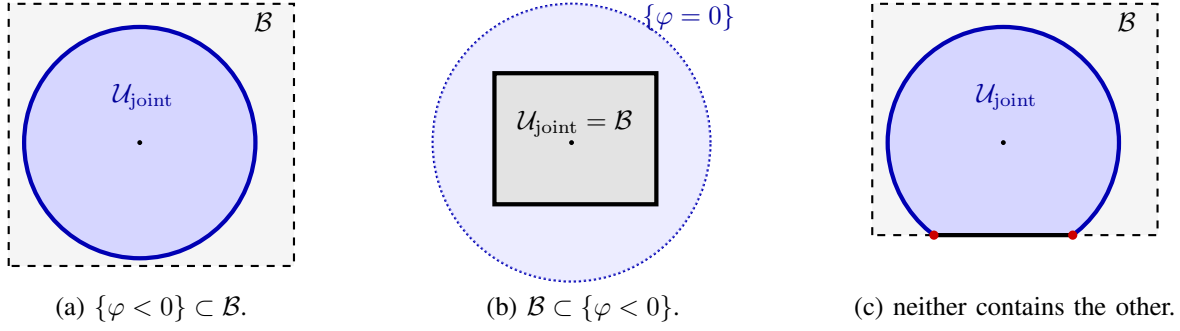
\begin{figure*}[!ht]
\centering
\begin{tikzpicture}[scale=1.02,>=Latex,
  panelnote/.style={align=center,text width=44mm}]
\begin{scope}
  \def\R{1.5}
  \fill[gray!8]             (-1.7,-1.6) rectangle (2.0,1.8);
  \draw[black,dashed,thick] (-1.7,-1.6) rectangle (2.0,1.8);
  \fill[blue!16] (0,0) circle (\R);
  \draw[blue!70!black,line width=1.6pt] (0,0) circle (\R);
  \fill (0,0) circle (0.9pt);
  \node[] at (1.60,1.57) {$\Ubox$};
  \node[blue!60!black] at (0.05,0.60) {$\Ujoint$};
  \node[] at (0.15,-2.15) {(a)~$\{\varphi<0\}\subset\Ubox$.};
\end{scope}
\begin{scope}[xshift=5.6cm]
  \def\R{1.8}
  \draw[blue!70!black,line width=0.9pt,densely dotted] (0,0) circle (\R);
  \fill[blue!16,opacity=.45] (0,0) circle (\R);
  \fill[gray!22]                (-1.0,-0.8) rectangle (1.1,0.9);
  \draw[black,line width=1.6pt] (-1.0,-0.8) rectangle (1.1,0.9);
  \fill (0,0) circle (0.9pt);
  \node[blue!60!black] at (1.52,1.62) {$\{\varphi=0\}$};
  \node[] at (0.05,0.25) {$\Ujoint=\Ubox$};
  \node[] at (0.15,-2.15) {(b)~$\Ubox\subset\{\varphi<0\}$.};
\end{scope}
\begin{scope}[xshift=11.2cm]
  \def\R{1.5}
  \fill[gray!8]             (-1.7,-1.2) rectangle (2.0,1.8);
  \draw[black,dashed,thick] (-1.7,-1.2) rectangle (2.0,1.8);
  \begin{scope}
    \clip (-1.7,-1.2) rectangle (2.0,1.8);
    \fill[blue!16] (0,0) circle (\R);
    \draw[blue!70!black,line width=1.6pt] (0,0) circle (\R);
  \end{scope}
  \draw[black,line width=1.6pt] (-0.9,-1.2) -- (0.9,-1.2);
  \fill[red!80!black] ( 0.9,-1.2) circle (1.7pt);
  \fill[red!80!black] (-0.9,-1.2) circle (1.7pt);
  \fill (0,0) circle (0.9pt);
  \node[] at (1.60,1.57) {$\Ubox$};
  \node[blue!60!black] at (0.05,0.60) {$\Ujoint$};
  \node[] at (0.15,-2.15) {(c)~neither contains the other.};
\end{scope}
\end{tikzpicture}
\caption{Possible configurations between $\Ubox$ and joint capacity constraint $\{\varphi<0\}$.}
\label{fig:geometries}
\end{figure*}

It is worth mentioning that a box $\Ubox$ and a smooth envelope $\{\varphi < 0\}$ can meet in three qualitatively distinct ways, shown in \Cref{fig:geometries}. If $\{\varphi < 0\} \subset \Ubox$, the joint constraint alone is active on $\partial\Ujoint$ and the channel factors $S^{(k)}(u_k)$ are bounded away from zero on $\Ujoint$. If $\Ubox \subset \{\varphi < 0\}$, only the box is active and the realization reduces to a channel-wise one, for which the anisotropic gain \eqref{eq:G_aniso} offers no advantage over \eqref{eq:G_iso}. If neither set contains the other, both constraints are active on distinct portions of $\partial\Ujoint$, and $\Ujoint$ carries corners at which a box face meets the envelope. The third configuration, \Cref{fig:geometries}(c), is the general one, and it contains the first two as a special case. It is also the configuration in which anisotropy is of consequence, since the normal and tangential directions of the envelope are then distinguished along a portion of the boundary on which the box remains inactive. \Cref{asm:phi} accordingly places no inclusion between $\Ubox$ and $\{\varphi < 0\}$, and the results of \Cref{sec:results} hold for all three.

\begin{assumption} \label{asm:init}
The initial conditions satisfy $\bu(0) \in \Ujoint$, and the initial tracking error $e(0) \triangleq y(0) - y_d(0)$ belongs to a compact set $\Omega_0 \subset \mathbb{R}$. Moreover, the reference trajectory $y_d$ is separated from the joint constraint boundary in the sense that $\inf_{t \geq 0} \{-\varphi(\bu^{*}(t))\} > 0$, where $\bu^{*}(t)$ is any ideal (unconstrained) control input that achieves perfect tracking $y(t) \equiv y_d(t)$.
\end{assumption}

The separation condition in \Cref{asm:init} is a compatibility requirement analogous to the APC compatibility condition in the single-input framework \cite{kumar2026admissibility}. It asserts that the tracking task is feasible in the sense that perfect tracking can, in principle, be accomplished within $\Ujoint$. Enforcing it as a standing assumption separates the question of feasibility, which depends on the reference trajectory and the system's input-output gain, from the question of constraint satisfaction during transients, which is addressed by the AJ-APIR dynamics \eqref{eq:AJ_APIR}. In practice, this condition can be verified offline for a given mission profile prior to deployment.

With the system \eqref{eq:system}, the joint admissible set \eqref{eq:joint_set}, and the AJ-APIR dynamics \eqref{eq:AJ_APIR} in
place, the control objectives of this work are formally stated as follows.

\begin{problem}[Joint Admissibility and Output Tracking] \label{prob:main}
Design the commanded input $\buc$ in \eqref{eq:AJ_APIR} such that the following conditions hold simultaneously for all $t \geq 0$:
\begin{enumerate}[label=\emph{(\roman*)}]
  \item \textbf{Joint admissibility:} $\bu(t) \in \Ujoint$ for all $t \geq 0$, i.e., the realized input remains strictly inside the joint admissible set at all times.
  \item \textbf{Asymptotic output tracking:} The tracking error $e(t) \triangleq y(t) - y_d(t)$ satisfies $e(t) \to 0$ as
    $t \to \infty$, with all closed-loop signals remaining uniformly bounded.
\end{enumerate}
\end{problem}

\section{Parametric AJ-APIR Design}\label{sec:results}
We propose the  parametrization of the gain functions in \eqref{eq:G_aniso} as
\begin{subequations}\label{eq:gains}
	\begin{align}
		G_{\perp}(\bu) &\triangleq
		p_{1,\perp}\,\bigl(-\varphi(\bu)\bigr)^{\gamma_\perp}\,\Sbox(\bu),
		\label{eq:Gperp}\\
		G_{\parallel}(\bu) &\triangleq
		p_{1,\parallel}\,\Sbox(\bu),
		\label{eq:Gpar}
	\end{align}
\end{subequations}
where $p_{1,\perp}, p_{1,\parallel}, \gamma_\perp > 0$ are design parameters, and
\begin{equation}
	\label{eq:Sbox}
	\Sbox(\bu) \triangleq \prod_{k=1}^m S^{(k)}(u_k)
\end{equation}
is the product of per-channel admissibility factors defined in \eqref{eq:Sk}. 

The product structure in \eqref{eq:gains} endows the AJ-APIR with a two-level constraint hierarchy. First, $\Sbox(\bu) \to 0$ whenever any $u_k$ approaches its box bound, so both $G_\perp$ and $G_\parallel$ vanish simultaneously. This mimics the single-input APIR, that is, for a commanded input that remains bounded as $\Sbox \to 0$, the realized velocity at a box face reduces to $\dot{\bu} = -p_2\bu$, which points strictly inward, so no such command can drive an individual actuator out of $\mathcal{U}_k$. Second, $(-\varphi(\bu))^{\gamma_\perp} \to 0$ as $\bu$ approaches the joint boundary from the interior, so $G_\perp \to 0$ while $G_\parallel$ have a positive value proportional to $\Sbox$. This is precisely the anisotropy that prevents joint constraint escape while preserving tangential authority. Near the joint boundary but away from any box boundary, $G_\parallel \approx p_{1,\parallel} > G_\perp \approx 0$, so the ratio $G_\parallel / G_\perp \to \infty$, the realizable-velocity ellipsoid becomes arbitrarily elongated along the tangent plane, which is discussed in \Cref{rem:spectral}.

The following theorem establishes admissibility of the realized input under a bounded commanded input, and constitutes the multi-input counterpart of \cite[Theorem 1]{kumar2026admissibility}.

\begin{theorem}\label{lem:FI}
Let \Cref{asm:model,asm:phi} hold, regard $G_\perp,G_\parallel$ be as in \eqref{eq:gains}, and put $\varphi_0 \triangleq -\varphi(\mathbf{0}) > 0$, $\varphi_{\sup} \triangleq \sup_{\Ujoint}(-\varphi) < \infty$, $L \triangleq \sup_{\Ujoint}\|\nabla\varphi\| < \infty$ and $C \triangleq \max\{p_{1,\perp}\varphi_{\sup}^{\gamma_\perp},\,p_{1,\parallel}\}$. Let the commanded input be piecewise continuous and satisfy $\|\buc(t)\| \leq \xi$ for all $t \geq 0$. For each channel $k$, let $u_{k,\xi}^{+} \in (0,u_{\max}^{(k)})$ and $u_{k,\xi}^{-} \in (u_{\min}^{(k)},0)$ denote the unique roots of
\begin{subequations}\label{eq:channelroots}
\begin{align}
C\xi\bigl[1-(u/u_{\max}^{(k)})^{\gamma^{(k)}}\bigr]&= p_2\,u,\label{eq:rootplus}\\
-C\xi\bigl[1-(u/u_{\min}^{(k)})^{\gamma^{(k)}}\bigr]&=p_2\,u,\label{eq:rootminus}
\end{align}
\end{subequations}
and set $\xi_\varphi(\xi) \triangleq \min\bigl\{\varphi_0/2,\ [\,p_2\varphi_0/(2L\,p_{1,\perp}\,\xi)\,]^{1/\gamma_\perp}\bigr\}$ and $\eta_{\xi,0} \triangleq \min\{-\varphi(\bu(0)),\,\xi_\varphi(\xi)\}$. Then, for every $\bu(0) \in \Ujoint$, the solution of \eqref{eq:AJ_APIR} exists for all $t \geq 0$ and satisfies $\bu(t) \in \mathcal{K}_{\xi,0}$, where
\begin{equation}\label{eq:Kxi0}
\begin{aligned}
{\mathcal{K}_{\xi,0}} &\;{\triangleq\; \Bigl(\textstyle\prod_{k=1}^{m}\mathcal{I}_{k,\xi}\Bigr) \cap\, \Ujoint^{\eta_{\xi,0}},}\\
{\mathcal{I}_{k,\xi}} &\;{\triangleq\; \bigl[\min\{u_k(0),u_{k,\xi}^{-}\},\ \max\{u_k(0),u_{k,\xi}^{+}\}\bigr].}
\end{aligned}
\end{equation}
{The set $\mathcal{K}_{\xi,0}$ is a compact subset of $\Ujoint$ containing $\bu(0)$; in particular $\Ujoint$ is forward invariant, and}
\begin{equation}\label{eq:Dxi0}\underline{D}_{\xi,0} \;\triangleq\; \min_{\bu \in \mathcal{K}_{\xi,0}}\ \min\{G_\perp(\bu),\,G_\parallel(\bu)\} \;>\; 0 .
\end{equation}
\end{theorem}

\begin{proof}
The proof establishes the two invariance properties separately, one for each level of the constraint hierarchy, and then combines them.

\textit{Box boundaries.}
{Fix $k$ and consider the branch $u_k \geq 0$, on which $\rho^{(k)} = 1$ in \eqref{eq:Sk}. Since $\|\bG\| = \max\{G_\perp,G_\parallel\} \leq C\,\Sbox \leq C\,S^{(k)}(u_k)$, every factor of \eqref{eq:Sbox} lying in $[0,1]$, the $k$-th component of \eqref{eq:AJ_APIR} obeys $\dot{u}_k \leq \|\bG\|\xi - p_2u_k \leq \Phi_k^{+}(u_k)$ with $\Phi_k^{+}(u) \triangleq C\xi[1-(u/u_{\max}^{(k)})^{\gamma^{(k)}}] - p_2u$. The map $\Phi_k^{+}$ is strictly decreasing on $[0,u_{\max}^{(k)}]$ and satisfies $\Phi_k^{+}(0) = C\xi > 0$ and $\Phi_k^{+}(u_{\max}^{(k)}) = -p_2u_{\max}^{(k)} < 0$, so \eqref{eq:rootplus} admits the unique root $u_{k,\xi}^{+}$, that root is interior to $(0,u_{\max}^{(k)})$, and $\dot{u}_k \leq 0$ at every $u_k \geq u_{k,\xi}^{+}$. On the branch $u_k \leq 0$ one has $\rho^{(k)} = 0$ and, symmetrically, $\dot{u}_k \geq \Phi_k^{-}(u_k)$ with $\Phi_k^{-}(u) \triangleq -C\xi[1-(u/u_{\min}^{(k)})^{\gamma^{(k)}}] - p_2u$, strictly decreasing on $[u_{\min}^{(k)},0]$ with $\Phi_k^{-}(u_{\min}^{(k)}) = -p_2u_{\min}^{(k)} > 0$ and $\Phi_k^{-}(0) = -C\xi < 0$, so \eqref{eq:rootminus} admits the unique root $u_{k,\xi}^{-} \in (u_{\min}^{(k)},0)$ and $\dot{u}_k \geq 0$ at every $u_k \leq u_{k,\xi}^{-}$. The vector field therefore points inward at both endpoints of $\mathcal{I}_{k,\xi}$, which is consequently positively invariant, and the same holds for every channel.}

\textit{Joint boundary.}
Compute $\dot{\varphi} = \nabla\varphi^{\top}\dot{\bu}$. Using $\nabla\varphi = \|\nabla\varphi\|\bn$ and the relation in  \eqref{eq:G_aniso} implies $\nabla\varphi^{\top}\bG\buc = \|\nabla\varphi\| G_\perp (\bn^{\top}\buc)$ (since $\bn^{\top}\bP = 0$). Therefore,
\begin{equation}
\label{eq:phidot}
\dot{\varphi}
= \|\nabla\varphi\|\,G_\perp(\bn^{\top}\buc)
- p_2\nabla\varphi^{\top}\bu .
\end{equation}

Evaluating the convexity inequality $\varphi(\mathbf{y}) \geq \varphi(\bu) + \nabla\varphi^{\top}(\bu)(\mathbf{y}-\bu)$ at $\mathbf{y} = \mathbf{0}$, without imposing $\varphi(\bu) = 0$, gives the global support bound
\begin{equation}\label{eq:strictsupport}
\nabla\varphi^{\top}(\bu)\bu \geq \varphi(\bu) - \varphi(\mathbf{0}) = \varphi_0 + \varphi(\bu), \; \bu \in \Ujoint.
\end{equation}
 On the region $\{-\varphi < \xi_\varphi(\xi)\}$ one has $-\varphi < \varphi_0/2$, hence $\varphi_0 + \varphi > \varphi_0/2$, while $\Sbox \leq 1$ and $(-\varphi)^{\gamma_\perp} < \xi_\varphi^{\gamma_\perp} \leq p_2\varphi_0/(2Lp_{1,\perp}\xi)$ bound the first term of \eqref{eq:phidot}. Substituting both into \eqref{eq:phidot} gives $\dot\varphi < p_2\varphi_0/2 - p_2\varphi_0/2 = 0$, so $-\varphi$ is strictly increasing there. Consequently $\{-\varphi \geq \mu\}$ is forward invariant for every $\mu < \eta_{\xi,0}$, and letting $\mu$ increase to $\eta_{\xi,0}$ gives $-\varphi(\bu(t)) \geq \eta_{\xi,0}$ for all $t \geq 0$.

\textit{Compactness and completeness.} The two invariance properties confine $\bu(t)$ to $\mathcal{K}_{\xi,0}$ on the maximal interval of existence. Each $\mathcal{I}_{k,\xi}$ is a compact subinterval of the open interval $\mathcal{U}_k$, because $u_k(0) \in \mathcal{U}_k$ and $u_{k,\xi}^{\pm}$ are interior, while $\{-\varphi \geq \eta_{\xi,0}\}$ is closed with $\eta_{\xi,0} > 0$ since $\varphi(\bu(0)) < 0$. Hence $\mathcal{K}_{\xi,0}$ is compact, contains $\bu(0)$, and lies in $\Ujoint$. Boundedness of the realized input together with local Lipschitzness of the right-hand side of \eqref{eq:AJ_APIR} away from the origin yields forward completeness, solutions being understood in the Carath\'eodory sense at $\bu = \mathbf{0}$, where $\bn$ is undefined. Finally, $\Sbox > 0$ on the compact product of channel intervals and $-\varphi \geq \eta_{\xi,0} > 0$ on the second factor, so the continuous functions $G_\perp$ and $G_\parallel$ of \eqref{eq:gains} are strictly positive on $\mathcal{K}_{\xi,0}$ and therefore attain positive minima there, which is \eqref{eq:Dxi0}.
\end{proof}
To achieve the desired tracking objective, the proposed AJ-APIR is interconnected with the nonlinear plant, yielding the following augmented system:
\begin{subequations}
\label{eq:aug_system}
\begin{align}
  \dot{x}_i &= f_i(\bar{x}_i) + g_i(\bar{x}_i)\,x_{i+1},
              \quad i = 1,\ldots,n-1, \label{eq:system_a1}\\
  \dot{x}_n &= f_n(\bar{x}_n) + \mathbf{g}_n^{\top}(\bar{x}_n)\,\bu,\label{eq:system_b1}
             \\
\dot {\bu}&= \bG(\bu)\bu_c-p_2\bu\label{eq:system_c1}
\end{align}
\end{subequations}

Since the AJ-APIR generates a continuously differentiable input signal, its output can be directly embedded into the recursive backstepping design.  Leveraging this property, the APC is developed within the
standard backstepping framework. To this end, we introduce the following error coordinates for the output-tracking objective $z_1 \triangleq y - y_d$, $z_i \triangleq x_i - \eta_{i-1}$ for $i = 2,\ldots,n$, where $\eta_{i-1}$ is the stabilizing function designed at the $(i{-}1)\textsuperscript{th}$ step, and $\varrho$ is the realization error defined at the $n\textsuperscript{th}$-step. We now proceed with backstepping design in a step-wise manner.

\textit{Step 1:}
Differentiating $z_1 = y - y_d$ along \eqref{eq:system} gives
\begin{equation}\label{eq:z1dot}
\dot{z}_1 = f_1(x_1) + g_1(x_1)\,x_2 - \dot{y}_d .
\end{equation}
Viewing $x_2$ as a virtual control and taking $V_1 \triangleq \tfrac{1}{2}z_1^2$, one may obtain $\dot{V}_1 = z_1[f_1 + g_1 x_2 - \dot{y}_d]$, which, on substituting the value of $x_2$ as $x_2 = z_2 + \eta_1$, becomes
\begin{equation}\label{eq:V1dot_1}
\dot{V}_1 = z_1\bigl[f_1 + g_1(z_2 + \eta_1) - \dot{y}_d\bigr].
\end{equation}
Designing the stabilizing function as
\begin{equation}\label{eq:eta1}
\eta_1 \triangleq \frac{1}{g_1(x_1)}
\bigl[\dot{y}_d - f_1(x_1) - k_1 z_1\bigr],
\ k_1 > 0,
\end{equation}
yields
\begin{equation}\label{eq:V1dot}
\dot{V}_1 = -k_1 z_1^2 + g_1(x_1)\,z_1 z_2 .
\end{equation}

\textit{Step $i$ $(2 \leq i \leq n-1)$:} Differentiating $z_i = x_i - \eta_{i-1}$ gives $\dot{z}_i = f_i(\bar{x}_i) + g_i(\bar{x}_i)x_{i+1} - \dot{\eta}_{i-1}$. Augmenting the Lyapunov function as $V_i \triangleq V_{i-1} + \tfrac{1}{2}z_i^2$ and using the result of the preceding step together with the fact that $x_{i+1} = z_{i+1} + \eta_i$, one may obtain
\begin{align}
\dot{V}_i =&~ -\sum_{j=1}^{i-1} k_j z_j^2
+ g_{i-1}(\bar{x}_{i-1})\,z_{i-1} z_i \nonumber\\
&~+ z_i\bigl[f_i + g_i(z_{i+1} + \eta_i) - \dot{\eta}_{i-1}\bigr].\label{eq:Vidot_1}
\end{align}
The stabilizing function is designed as
\begin{equation}
\label{eq:virtual}
\eta_i \triangleq \frac{\dot{\eta}_{i-1} - f_i(\bar{x}_i)
- g_{i-1}(\bar{x}_{i-1})\,z_{i-1} - k_i z_i}{g_i(\bar{x}_i)},
\ k_i > 0.
\end{equation}
By substituting the value of $\eta_i$ from \eqref{eq:virtual} into \eqref{eq:Vidot_1}, we obtain
\begin{equation}\label{eq:Vidot}
\dot{V}_i = -\sum_{j=1}^{i} k_j z_j^2 + g_i(\bar{x}_i)\,z_i z_{i+1}.
\end{equation}
Note that $\dot{\eta}_{i-1}$ depends on $\bar{x}_i$ and $y_d^{(i)}$ through the chain rule and is bounded under \Cref{asm:reference,asm:model} whenever the state evolves in a compact set.

\textit{Step $n$:} At this step, we define $v\triangleq \bg_n^{\top}(\bar{x}_n)\,\bu$, so that $\dot{z}_n = f_n(\bar{x}_n) + v - \dot{\eta}_{n-1}$. With $V_n \triangleq V_{n-1} + \tfrac{1}{2}z_n^2$ and \eqref{eq:Vidot}, the time derivative of $V_n$ becomes
\begin{align}
\dot{V}_n =&~ -\sum_{j=1}^{n-1} k_j z_j^2 + g_{n-1}(\bar{x}_{n-1})z_{n-1}z_n\nonumber\\
&~+ z_n\left[f_n(\bar{x}_n) + v - \dot{\eta}_{n-1}\right]. \label{eq:Vndot_1}
\end{align}
Unlike a standard backstepping step, $v$ cannot be assigned its stabilizing value directly, since $v$ is produced by the AJ-APIR \eqref{eq:AJ_APIR} and carries dynamics of its own. We therefore introduce a realization error as
\begin{equation}\label{eq:rho}
\varrho \triangleq v - \eta_n ,
\end{equation}
where $\eta_n$ is the stabilizing function. If we design $\eta_n$ as
\begin{equation}
\label{eq:alphan}
\eta_n \triangleq
\dot{\eta}_{n-1} - f_n(\bar{x}_n)
- g_{n-1}(\bar{x}_{n-1})\,z_{n-1} - k_n z_n,
\ k_n > 0,
\end{equation}
then the derivative of the Lyapunov function candidate given in \eqref{eq:Vndot_1} using the relation $v = \varrho + \eta_n$ becomes
\begin{equation}\label{eq:Vndot}
\dot{V}_n = -\sum_{j=1}^{n} k_j z_j^2 + z_n \varrho .
\end{equation}

\textit{Step $n{+}1$:} We synthesize the commanded input $\buc$ so that $\varrho$ also converges to zero, while the joint admissibility condition is preserved. Differentiating \eqref{eq:rho} and noticing that $\bg_n$ is itself state-dependent results in $\dot{\varrho} = \dot{\bg}_n^{\top}\bu + \bg_n^{\top}\dot{\bu} - \dot{\eta}_n$, which, on substituting the AJ-APIR dynamics \eqref{eq:AJ_APIR}, becomes
\begin{equation}\label{eq:rhodot_1}
\dot{\varrho} = \dot{\bg}_n^{\top}\bu
+ \bg_n^{\top}\bG(\bu)\,\buc - p_2 v - \dot{\eta}_n ,
\end{equation}
where $\dot{\bg}_n = \bigl[\partial \bg_n/\partial \bar{x}_n\bigr] \dot{\bar{x}}_n$ is computable from the model and the measured state. Choose the final Lyapunov function candidate $V$ as
\begin{equation} \label{eq:Lyap}
V \triangleq \frac{1}{2}\sum_{i=1}^{n} z_i^2 + \frac{1}{2}\varrho^2.
\end{equation}
The time derivative of \eqref{eq:Lyap} can be obtained using \eqref{eq:Vndot} and \eqref{eq:rhodot_1} as
\begin{equation}\label{eq:Vdot_1}
\dot{V} = -\sum_{j=1}^{n} k_j z_j^2 + z_n\varrho
+ \varrho\bigl[\dot{\bg}_n^{\top}\bu + \bg_n^{\top}\bG\buc
- p_2 v - \dot{\eta}_n\bigr].
\end{equation}
We define a scalar commanded effective input $u_c^*$ as $u_c^* := \bg_n^{\top}(\bar{x}_n)\bG(\bu)\buc$ and choose it as
\begin{equation} \label{eq:uceff}
u_c^* \triangleq
\dot{\eta}_n + p_2\,\eta_n - \dot{\bg}_n^{\top}\bu
- k_\varrho \varrho - z_n,
\ k_\varrho > 0.
\end{equation}
Note that $u_c^*$ should be realized through the relation
\begin{equation}\label{eq:ucrel}
\bg_n^{\top}(\bar{x}_n)\,\bG(\bu)\,\buc = u_c^*.
\end{equation}

Since $\bu \in \mathbb{R}^m$ with $m > 1$, the scalar relation \eqref{eq:ucrel} does not determine $\buc$ uniquely, and the residual freedom is resolved by selecting the minimum-energy solution in the metric induced by $\bG$. Substituting $\bq = \bG^{1/2}\buc$ converts $\min \|\bG^{1/2}\buc\|^2$ subject to \eqref{eq:ucrel} into the classical minimum-norm problem $\min\|\bq\|^2$ subject to $(\bG^{1/2}\bg_n)^{\top}\bq = u_c^*$, whose solution \cite{johansen2013control} $\bq^\star = \bG^{1/2}\bg_n u_c^*/(\bg_n^{\top}\bG\bg_n)$ is, upon multiplication by $\bG^{-1/2}$,
\begin{equation}\label{eq:uc}
\buc \triangleq \frac{\bg_n(\bar{x}_n)} {\bg_n^{\top}(\bar{x}_n)\,\bG(\bu)\,\bg_n(\bar{x}_n)}u_c^*.
\end{equation}
Using \eqref{eq:uceff}--\eqref{eq:uc}, the expression in \eqref{eq:Vdot_1} reduces to
\begin{equation}\label{eq:Vdot}
\dot{V} = -\sum_{j=1}^{n} k_j z_j^2 - (k_\varrho + p_2)\,\varrho^2\leq -2\lambda V,
\end{equation}
with $\lambda \triangleq \min\{k_1,\ldots,k_n,\,k_\varrho + p_2\}$, which is negative definite on every domain on which the commanded-input map is well defined.

\begin{remark} \label{rem:uc_structure}
The commanded law \eqref{eq:uceff}--\eqref{eq:uc} is the multi-input counterpart of the single-input construction of \cite{kumar2026admissibility}. For $m = 1$ the projector $\bP$ is void, $\bG$ reduces to a scalar value, the term $\dot{\bg}_n^{\top}\bu$ is absorbed into the term $\dot{\eta}_n$, and \eqref{eq:uc} reduces to the control law given in \cite{kumar2026admissibility}. Thus, the present design generalizes its single-input to a multi-input framework. It is also worth noting that the extension from the single-input to the multi-input setting introduces a substantive change in the realization structure. For $m>1$, the scalar gain in the single-input case is replaced by the quadratic form $\bg_n^{\top}\bG\bg_n$. Consequently, it is this quadratic form, rather than $\bG$ itself, that determines the conditioning of the realization channel and must therefore remain bounded away from zero.
\end{remark}
\begin{lemma}\label{lem:cond}
Let \Cref{asm:model,asm:phi} hold, and $\mathcal{K}_{\xi,0}$ and $\underline{D}_{\xi,0}$ be as in \Cref{lem:FI}. Then, for every $\bu \in \mathcal{K}_{\xi,0}$,
\begin{equation}
\label{eq:denombound}
\bg_n^{\top}\bG(\bu)\,\bg_n
\;\geq\;
\underline{D}_{\xi,0} \|\bg_n\|^2
\;\geq\;
\underline{D}_{\xi,0}\, g_{\min}^2 \;>\; 0.
\end{equation}
\end{lemma}
	
\begin{proof}
From \eqref{eq:G_aniso}, $\bG$ is symmetric with eigenvalue $G_\perp$ along $\bn(\bu)$ and eigenvalue $G_\parallel$ on the tangent complement, whence $\bG \succeq \min\{G_\perp,G_\parallel\}\bI_m$ and therefore $\bg_n^{\top}\bG\bg_n \geq \min\{G_\perp,G_\parallel\}\|\bg_n\|^2$. On $\mathcal{K}_{\xi,0}$ one has $\min\{G_\perp,G_\parallel\} \geq \underline{D}_{\xi,0}$ by \eqref{eq:Dxi0}, which is the first inequality. The last inequality follows from $\|\bg_n\| \geq g_{\min}$ (\Cref{asm:model}).
\end{proof}

Let $\mathcal{N}(\mathbf{z},\varrho,\bu,t) \triangleq \|\bg_n\|\,|u_c^*|$ denote the command numerator evaluated along \eqref{eq:uceff}, and, for $c > 0$, let $\Omega_c \triangleq \{(\mathbf{z},\varrho) : V \leq c\}$ be the sublevel set of \eqref{eq:Lyap}.

\begin{definition}[AJ-APC compatibility]\label{def:compat}
The level $c > 0$ is \emph{AJ-APC compatible} if the zero-error realization $\eta_n^*(t) \triangleq \eta_n|_{\mathbf{z}=\mathbf{0},\varrho=0}$ is realizable by an input remaining in a compact subset of $\mathrm{int}(\Ujoint)$, and there exists $\xi > 0$ such that
\begin{equation}\label{eq:self_consistency}
Q_c(\xi) \triangleq\sup_{\substack{(\mathbf{z},\varrho) \in \Omega_c,\;\bu \in \mathcal{K}_{\xi,0},\; t \geq 0}}\frac{\mathcal{N}(\mathbf{z},\varrho,\bu,t)}{\bg_n^{\top}\bG(\bu)\bg_n}\;\leq\; \xi .
\end{equation}
\end{definition}
The ratio in \eqref{eq:self_consistency} is well defined by \Cref{lem:cond}, which also yields a more conservative but often simpler sufficient test. Let $M_c \triangleq \sup\{\mathcal{N}(\mathbf{z},\varrho,\bu,t) : (\mathbf{z},\varrho) \in \Omega_c,\ \bu \in \Ujoint,\ t \geq 0\}$ denote the supremum of the numerator. Since $\bg_n^{\top}\bG\bg_n \geq \underline{D}_{\xi,0}g_{\min}^2$ on $\mathcal{K}_{\xi,0}$, one has $Q_c(\xi) \leq M_c/(\underline{D}_{\xi,0}g_{\min}^2)$, so that
\begin{equation}\label{eq:explicit_compat}
    M_c \;\leq\; \xi\,\underline{D}_{\xi,0}\,g_{\min}^2
\end{equation}
which is sufficient for \eqref{eq:self_consistency}. Since both sides of \eqref{eq:explicit_compat} are computable offline, $M_c(\xi)$ being finite by \Cref{lem:recursive_boundedness}, so a prescribed level $c$ can be certified directly. Condition \eqref{eq:self_consistency} is the joint-constraint counterpart of the APC compatibility condition of \cite{kumar2026admissibility} and plays the same role, namely that of coupling the admissible initial-error region and the reference demand to the finite actuator authority and to the conditioning of the realization channel.

Two properties are needed before \Cref{def:compat} can be used. The first is that the numerator of \eqref{eq:self_consistency} is finite on every sublevel, so that $Q_c(\xi)$ is well defined. The second is that the condition is met by some sublevel, so that the class of compatible levels is nonempty. These are established next.

\begin{lemma}[Bounded recursive quantities on compatible sublevels]
\label{lem:recursive_boundedness}
Let \Cref{asm:reference,asm:model,asm:phi,asm:init} hold. Then all stabilizing functions $\eta_1,\ldots,\eta_{n-1}$, the desired effective input $\eta_n$, the derivatives required by the recursion \eqref{eq:eta1}--\eqref{eq:uceff}, and the command numerator $\mathcal{N}$ are uniformly bounded on $\Omega_c \times \Ujoint \times [0,\infty)$ for every $c > 0$.
\end{lemma}

\begin{proof}
We establish the boundedness recursively over the backstepping steps. On $\Omega_c$, the error vector is bounded. By \Cref{asm:reference}, the reference trajectory and its derivatives up to order $n+1$ are also uniformly bounded. Hence, $x_1=z_1+y_d$ evolves in a compact set.

It follows from \eqref{eq:eta1} that $\eta_1$ is bounded on this set. Moreover, the derivatives of $\eta_1$ required at the subsequent backstepping steps are bounded because $f_1$ and $g_1$ are continuously differentiable and $g_1\neq 0$ by \Cref{asm:model}. Since $g_1$ is continuous and nonzero on the resulting compact set, there exists a constant $c_1>0$ such that $\lvert g_1\rvert\ge c_1$ on that set.

Now suppose that the claim holds through Step $i-1$. Then $x_i=z_i+\eta_{i-1}$ is bounded and therefore evolves in a compact set. By the continuous differentiability of $f_i$ and $g_i$, together with $g_i\neq0$ from \Cref{asm:model}, there exists $c_i>0$ such that $\lvert g_i\rvert\ge c_i$ on this compact set. Consequently, $\eta_i$ in \eqref{eq:virtual} is bounded, as are all derivatives required by the subsequent recursion. Continuing this induction through the final step establishes the boundedness of $\eta_n$ in \eqref{eq:alphan} and of $\dot{\eta}_n$.

It remains to bound the two terms that are specific to the multi-input setting. Since $\bu \in \Ujoint$ and $\Ujoint$ is bounded by \Cref{asm:phi}, its closure being compact, $\bu$ is bounded, and $\bg_n$ is bounded on the compact state set by  \Cref{asm:model}. Therefore, $v = \bg_n^{\top}\bu$ is bounded and, with it, $\dot{x}_n = f_n + v$ and the remaining components $\dot{x}_i = f_i + g_i x_{i+1}$ of $\dot{\bar{x}}_n$. Continuous differentiability of $\bg_n$ then bounds $\dot{\bg}_n = [\partial\bg_n/\partial\bar{x}_n]\dot{\bar{x}}_n$ and therefore the drift term $\dot{\bg}_n^{\top}\bu$ in \eqref{eq:uceff}. Boundedness of $u_c^*$, and hence of $\mathcal{N} = \|\bg_n\||u_c^*|$, follows.
\end{proof}

\begin{lemma}[Local nonvacuity of AJ-APC compatibility]
\label{prop:compatibility_nonvacuity}
Let $\mathcal{K}_{\xi_0,0}$ be the compact set of \Cref{lem:FI} associated with a command bound $\xi_0 > 0$ and the initial input $\bu(0)$, and suppose that
\begin{equation} \label{eq:zero_error_compatibility}
    \sup_{t \geq 0,\; \bu \in \mathcal{K}_{\xi_0,0}}\;
    \frac{\mathcal{N}(\mathbf{0},0,\bu,t)}
    {\bg_n^{\top}\bG(\bu)\,\bg_n}
    \;<\; \xi_0 .
\end{equation}
Then there exists $c^\star > 0$ such that every sublevel $\Omega_c$ with $0 < c \leq c^\star$ is AJ-APC compatible, with command bound $\xi = \xi_0$.
\end{lemma}

\begin{proof}
By \Cref{lem:FI} the set $\mathcal{K}_{\xi_0,0}$ is a compact subset of $\Ujoint$, and by \Cref{lem:cond} the denominator satisfies $\bg_n^{\top}\bG\bg_n \geq \underline{D}_{\xi_0,0}g_{\min}^2 > 0$ on it. By \Cref{lem:recursive_boundedness} the numerator $\mathcal{N}$ is bounded on $\Omega_c \times \Ujoint \times [0,\infty)$ and, being a composition of continuously differentiable functions of the errors, the input and the reference, is continuous in $(\mathbf{z},\varrho)$ uniformly over $\mathcal{K}_{\xi_0,0}$ and over $t$, the reference and its derivatives being uniformly bounded by \Cref{asm:reference}. The quotient $\mathcal{N}/(\bg_n^{\top}\bG\bg_n)$ inherits both properties.

Write $q_0$ for the supremum in \eqref{eq:zero_error_compatibility}, so that $q_0 < \xi_0$, and note that $Q_c(\xi_0)$ is the supremum of the same quotient over the larger set $\Omega_c \times \mathcal{K}_{\xi_0,0} \times [0,\infty)$. Uniform continuity in $(\mathbf{z},\varrho)$ gives $Q_c(\xi_0) \to q_0$ as $c \to 0$, since $\Omega_c$ shrinks to $(\mathbf{0},0)$. Choosing $c^\star$ so that $Q_{c^\star}(\xi_0) \leq \xi_0$, which is possible because $q_0 < \xi_0$ strictly, and noting that $Q_c$ is non-decreasing in $c$, yields $Q_c(\xi_0) \leq \xi_0$ for every $c \leq c^\star$. This is \eqref{eq:self_consistency} with $\xi = \xi_0$.
\end{proof}

Note that condition \eqref{eq:zero_error_compatibility} requires only that the reference be trackable at zero error within the authority available on $\mathcal{K}_{\xi_0,0}$. It involves neither the closed-loop trajectory nor the sublevel $c$, and remains verifiable offline once $\xi_0$ and $\bu(0)$ are fixed.
\begin{theorem}
\label{thm:main}
Let \Cref{asm:reference,asm:model,asm:phi,asm:init} hold, let the level $c$ be AJ-APC compatible in the sense of \Cref{def:compat} with associated command bound $\xi > 0$, and consider the AJ-APIR \eqref{eq:AJ_APIR}--\eqref{eq:gains} driven by the commanded input \eqref{eq:eta1}--\eqref{eq:uc}. Then, for every initial condition with $V(0) \leq c$ and $\bu(0) \in \Ujoint$, the closed-loop solution exists for all $t \geq 0$ and:
\begin{enumerate}[label=\emph{(\roman*)}]
    \item $\bu(t) \in \mathcal{K}_{\xi,0} \subset \Ujoint$ for all $t \geq 0$;
    \item all closed-loop signals remain uniformly bounded;
    \item the tracking and realization errors converge exponentially,
    \begin{equation}
        \label{eq:expbound}
        \textstyle
        \sum_{i=1}^{n} z_i^2(t) + \varrho^2(t)
        \;\leq\;
        \bigl[\sum_{i=1}^{n} z_i^2(0) + \varrho^2(0)\bigr]\,
        e^{-2\lambda t},
    \end{equation}
    with $\lambda = \min\{k_1,\ldots,k_n,k_\varrho + p_2\}$, so that $y(t) \to y_d(t)$ exponentially;
    \item the commanded input satisfies $\|\buc(t)\| \leq \xi$ for all $t \geq 0$.
\end{enumerate}
\end{theorem}

\begin{proof}
The argument is a continuation on the maximal interval $[0,T)$ over which $\|\buc(t)\| \leq \xi$. This interval is nonempty: $\bu(0) \in \mathcal{K}_{\xi,0}$ by construction of $\mathcal{K}_{\xi,0}$ in \Cref{lem:FI} and $(\mathbf{z}(0),\varrho(0)) \in \Omega_c$ by hypothesis, so \eqref{eq:self_consistency} gives $\|\buc(0)\| \leq Q_c(\xi) \leq \xi$, and $\buc$ is continuous at $t=0$.

On $[0,T)$ the commanded input is bounded, so \Cref{lem:FI} applies and gives $\bu(t) \in \mathcal{K}_{\xi,0} \subset \Ujoint$, which establishes (i). From \Cref{lem:cond}, we have $\bg_n^{\top}\bG\bg_n \geq \underline{D}_{\xi,0} g_{\min}^2 > 0$, so that \eqref{eq:uc} is well defined and \eqref{eq:ucrel} holds exactly. The recursion of backstepping and \eqref{eq:Vdot} gives $\dot{V} \leq -2\lambda V$, hence the comparison lemma yields $V(t) \leq V(0)e^{-2\lambda t} \leq c$. In particular $\Omega_c$ is invariant, and \eqref{eq:expbound} follows on recalling the definition \eqref{eq:Lyap} of $V$, which establishes (iii).

It remains to verify that the assumed command bound is preserved, so that $T \to \infty$. Since $(\mathbf{z}(t),\varrho(t)) \in \Omega_c$ and $\bu(t) \in \mathcal{K}_{\xi,0}$ on $[0,T)$, the compatibility condition \eqref{eq:self_consistency} gives
\begin{equation*}
    \|\buc(t)\| = \frac{\mathcal{N}(\mathbf{z},\varrho,\bu,t)}{\bg_n^{\top}\bG(\bu)\bg_n} \;\leq\; Q_c(\xi) \;\leq\; \xi ,
\end{equation*}
so the bound holds with the same constant at the endpoint of every such interval. A standard continuation argument, therefore, extends the solution to all $t \geq 0$ and yields (iv).

Condition (ii) is then immediate: $V$ is non-increasing, so $\mathbf{z}$ and $\varrho$ remain within the compact sublevel set $\Omega_c$. The states $x_i$ are bounded because $x_i = z_i + \eta_{i-1}$ and each $\eta_{i-1}$ is a continuous function of bounded arguments by \Cref{asm:reference,asm:model}. The input vector $\bu$ is bounded by (i) since $\mathcal{K}_{\xi,0}$ is compact and $\buc$ is bounded by (iv).
\end{proof}

\begin{remark}
\label{rem:exact}
\Cref{thm:main} is a regional result, in the sense that it applies to AJ-APC compatible sublevels, and within that region, the convergence is exponential with a rate that is entirely a design quantity. One may observe from \eqref{eq:Vdot} that $\lambda$ contains $k_\varrho + p_2$ rather than $k_\varrho$ alone, so that the decay of the AJ-APIR filter contributes positively to the realization bandwidth instead of opposing it, the term $p_2\eta_n$ in \eqref{eq:uceff} having already absorbed its steady-state effect. In applications (for instance, when $y_d$ is generated online by an outer guidance loop \cite{kumar2024robust,kumar2024three}) where certifying \eqref{eq:self_consistency} offline is inconvenient, replacing the denominator of \eqref{eq:uc} by $\bg_n^{\top}\bG\bg_n + \epsilon_c$ with $\epsilon_c > 0$ yields the unconditional bound $\|\buc\| \leq \|\bg_n\||u_c^*|/\epsilon_c$ and retains forward invariance, at the price of a realization deficit of order $\epsilon_c$ in \eqref{eq:Vdot} and hence of ultimate boundedness in place of \eqref{eq:expbound}.
\end{remark}

\section{Numerical Simulation}\label{sec:sim}
In this section, we evaluate the performance of the proposed AJ-APIR framework through two numerical studies. The first study compares the proposed anisotropic realization with the isotropic realization in \eqref{eq:isotropic_APIR} under simultaneous individual box and joint capacity constraints. We first consider an arbitrary commanded input $\buc$ for a representative three-input system. For better visualization, animations for our results are available at \href{https://youtu.be/21RK1Gd4fkY}{https://youtu.be/21RK1Gd4fkY}.

For AJ-APIR, parameters are selected as $p_{1,\perp} =p_{1,\parallel}= 3$,  $p_2 = 2$, $\gamma^{(k)} = 2$, $\gamma_\perp = 1$, whereas the gains for the isotropic realization are $p_1^k=3$, $\gamma^{(k)} =p_2= 2$. The individual box constraints are specified as $u_1\in(-1.7,1.8)$, $u_2\in(-2.1,2.3)$, and $u_3\in(-1.6,1.9)$, while the joint capacity constraint is defined by $\varphi(\bu)=u_1^2+2u_2^2+0.5(u_3-0.15)^2-R^2<0$, with $R=1.6$. The commanded input is chosen as $\buc(t)=3.5[\sin(t),\ \cos(0.8t),\ \sin(1.3t+1)]^\top$, with initial condition $\bu_0=[0.5,0.5,0.5]^\top$.

\begin{figure*}[!ht]
    \centering
    \begin{subfigure}{0.45\linewidth}
        \centering
        \includegraphics[width=0.9\linewidth]{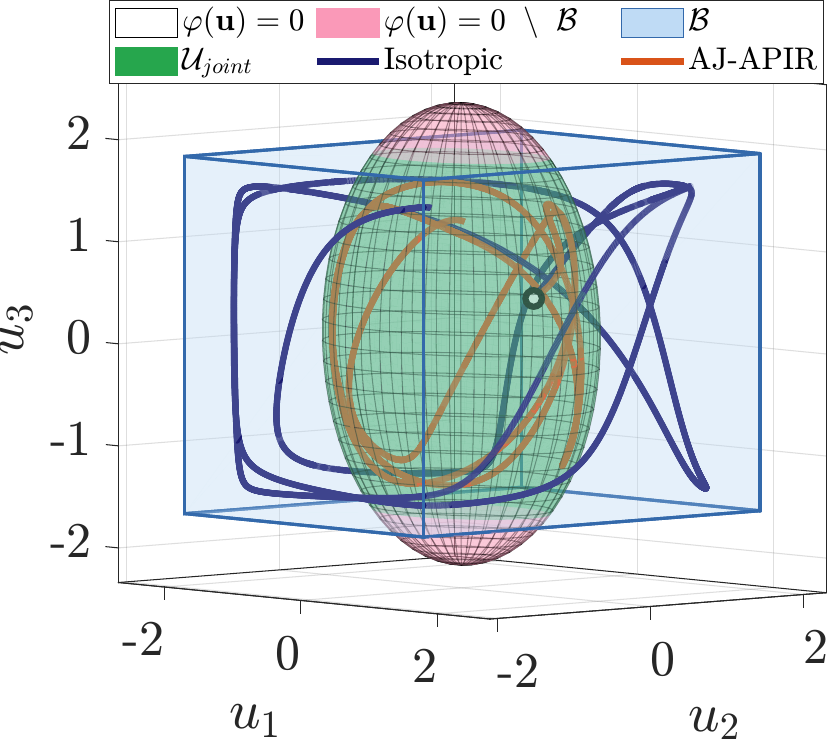}
        \caption{Input space.}
        \label{fig:phase_comp}
    \end{subfigure}%
\begin{subfigure}{0.45\linewidth}
        \centering
        \includegraphics[width=\linewidth]{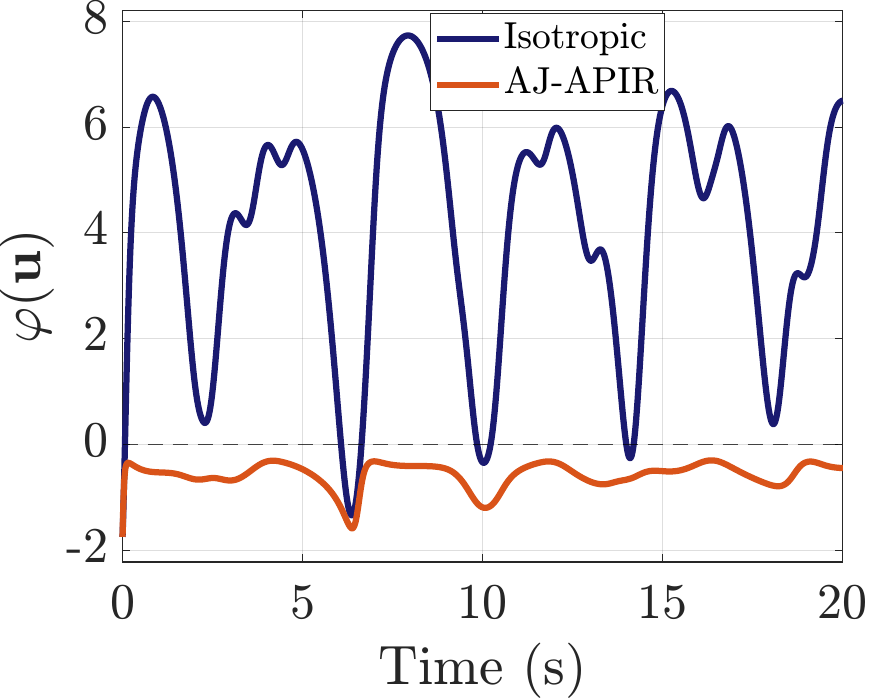}
        \caption{Joint Capacity Constraint \eqref{eq:joint_constraint}.}
        \label{fig:joint_Comp}
    \end{subfigure}   
        \begin{subfigure}{0.45\linewidth}
        \centering
        \includegraphics[width=\linewidth]{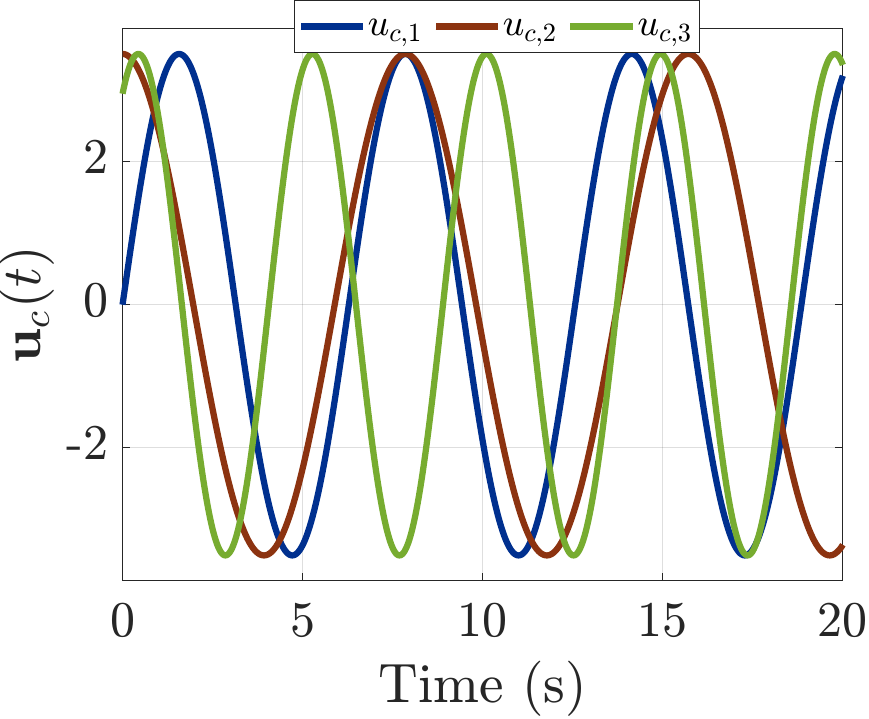}
        \caption{Commanded Input.}
        \label{fig:uc_comp}
    \end{subfigure}%
    \begin{subfigure}{0.45\linewidth}
        \centering
        \includegraphics[width=\linewidth]{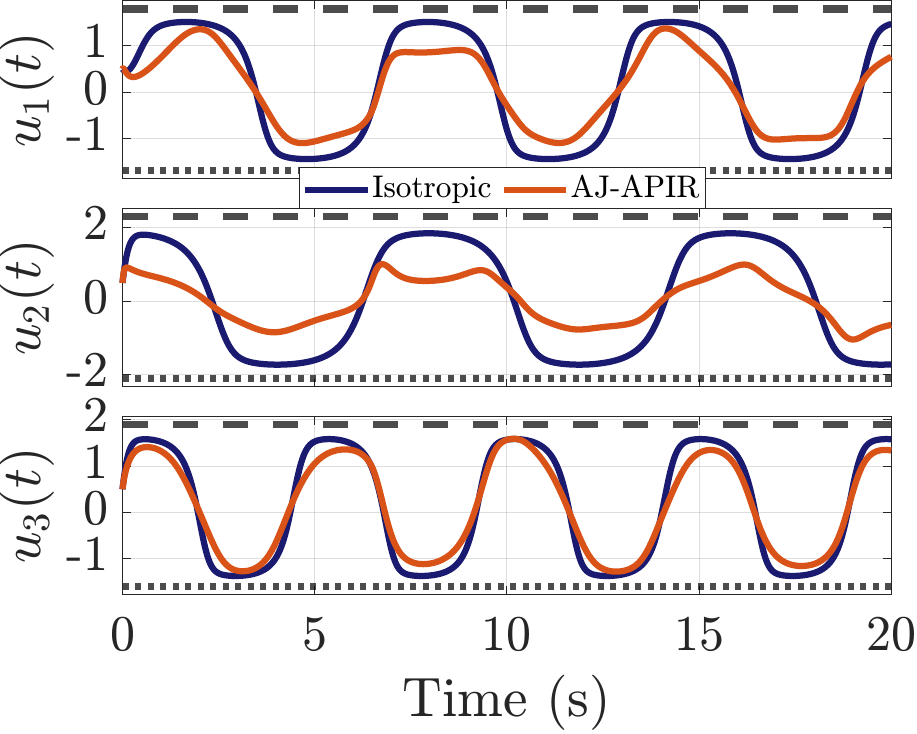}
        \caption{Realized Input.}
        \label{fig:u_comp}
    \end{subfigure}
\caption{Comparison of the proposed AJ-APIR Framework with \eqref{eq:isotropic_APIR} under individual and joint constraints.}
    \label{fig:comp_results}
\end{figure*}
The comparative results are presented in \Cref{fig:comp_results}. The phase portrait in \Cref{fig:phase_comp} illustrates that the isotropic realization trajectory repeatedly leaves the joint admissible set $\mathcal{U}_{\mathrm{joint}}$, whereas the AJ-APIR trajectory remains confined within the admissible region. This is also quantified by the joint constraint function $\varphi(\mathbf{u})$ shown in \Cref{fig:joint_Comp}. One may observe from \Cref{fig:joint_Comp} that the isotropic realization violates the joint constraint, while AJ-APIR maintains $\varphi(\mathbf{u}) <0$ throughout the time. As shown in \Cref{fig:uc_comp}, the commanded input $\mathbf{u}_c(t)$ exceeds the individual box limits. However, \Cref{fig:u_comp} confirms that both realizations keep every input channel within its respective box constraint, that is, $u_{\min,k} \leq u_k \leq u_{\max,k}$ for all time. Thus, satisfying the individual actuator limits alone does not, in general, guarantee satisfaction of the joint admissibility constraint. This, in turn, demonstrates that AJ-APIR enforces the individual and joint constraints simultaneously, whereas the channel-wise isotropic realization does not account for the coupling imposed by the shared capacity envelope.

For the output tracking case, we consider a second-order, two-input ($n{=}2$, $m{=}2$) strict-feedback system, given by
\begin{equation}
  \label{eq:sim_sys}
  \dot{x}_1 = x_2,\,
  \dot{x}_2 = f_2(x_1,x_2) + \mathbf{g}_2^{\top}(x_1,x_2)\bu,
\end{equation}
with $f_2 = -x_1 - \sin(x_2)$, $\mathbf{g}_2 = [1.5+0.3\sin(x_1),\;0.25+0.2\cos(x_2)]^{\top}$, $y = x_1$, and reference $y_d(t) = 4\sin(t)$. The individual box constraints are $u_1\in(-1.7,2.0)$ and $u_2\in(-1.2,1.8)$, while the joint capacity constraint is $\varphi(\bu)=u_1^2+u_2^2-R^2<0$, with $R=1.5$.

The AJ-APIR parameters are selected as $p_{1,\perp} =p_{1,\parallel}= 3$,  $p_2 = 2$, $\gamma^{(k)} = 2$, $\gamma_\perp = 1$, with backstepping gains $k_1 = k_2 = k_\rho= 2$. The initial conditions are chosen as $\mathbf{x}(0) = [1.2, -0.8]^\top$ and $\bu(0) = [0.4, -0.3]^\top$.
 \begin{figure*}[!ht]
    \centering
    \begin{subfigure}{0.45\linewidth}
        \centering
        \includegraphics[width=\linewidth]{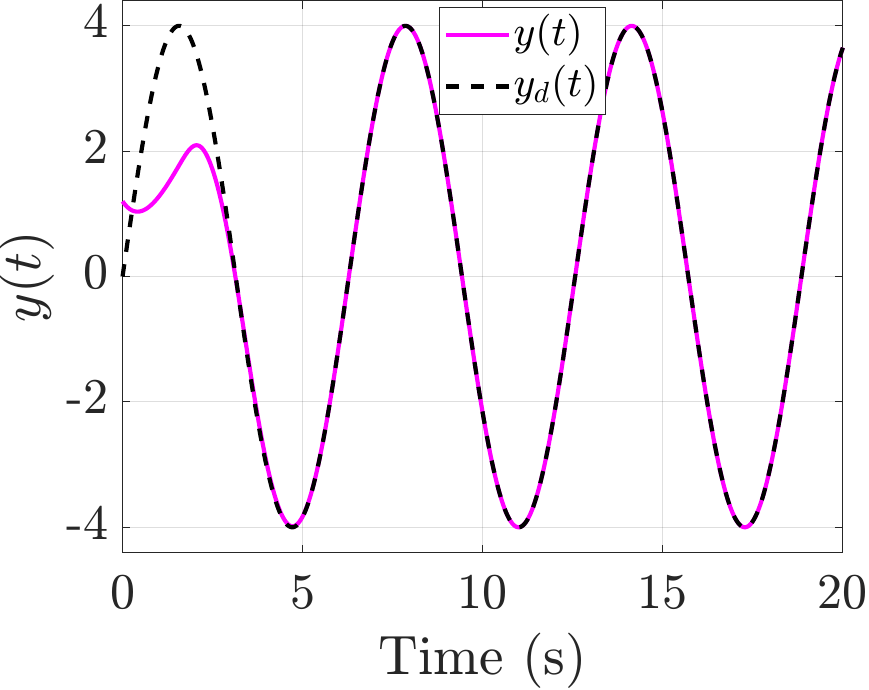}
        \caption{Trajectory.}
        \label{fig:traj1}
    \end{subfigure}%
        \begin{subfigure}{0.45\linewidth}
        \centering
        \includegraphics[width=\linewidth]{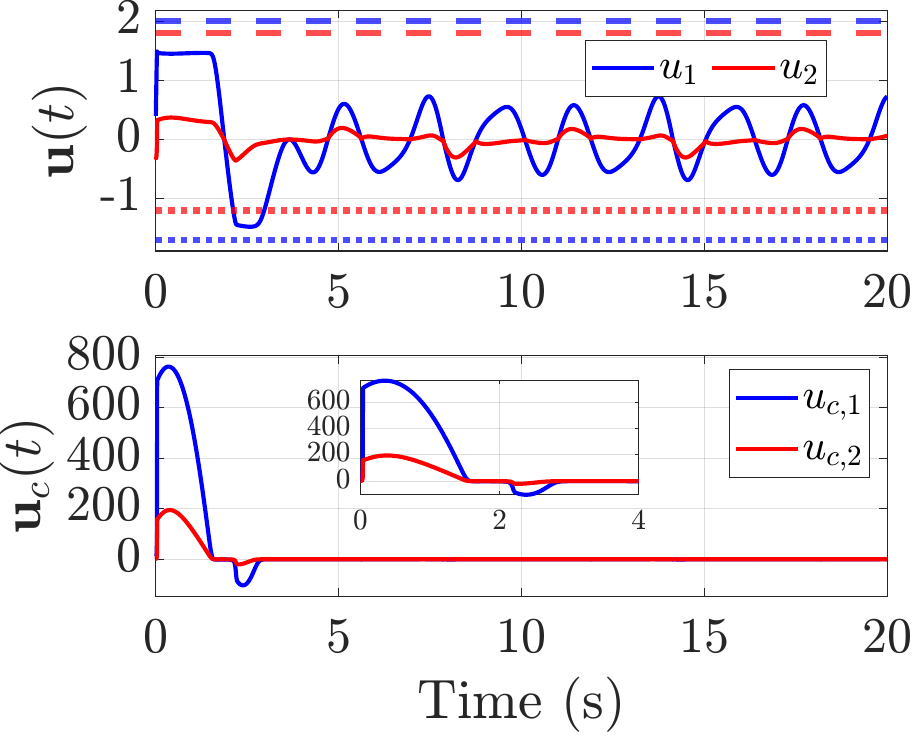}
        \caption{Realized and commanded input.}
        \label{fig:u1}
    \end{subfigure}
    \begin{subfigure}{0.45\linewidth}
        \centering
        \includegraphics[width=\linewidth]{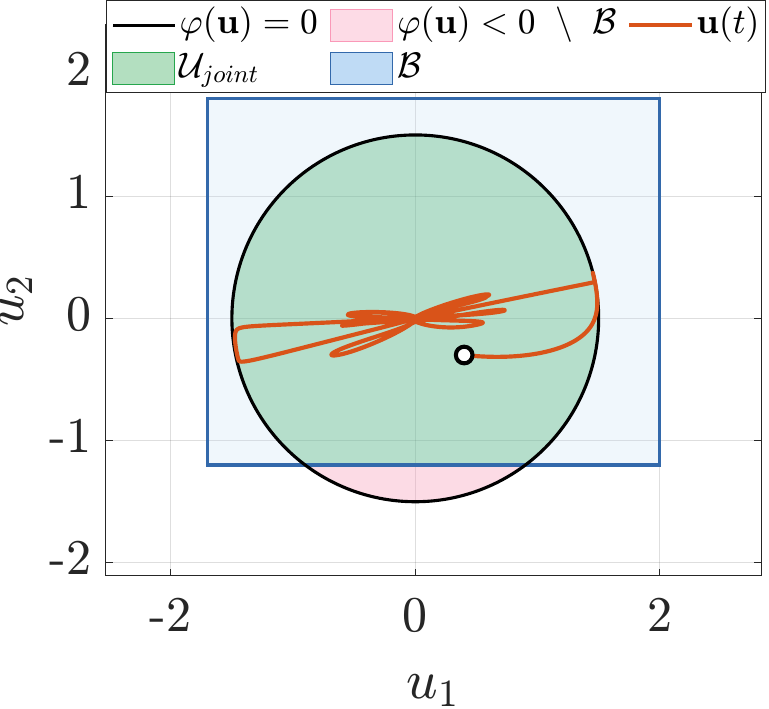}
        \caption{Input space.}
        \label{fig:phase1}
    \end{subfigure}%
           \begin{subfigure}{0.45\linewidth}
        \centering
        \includegraphics[width=\linewidth]{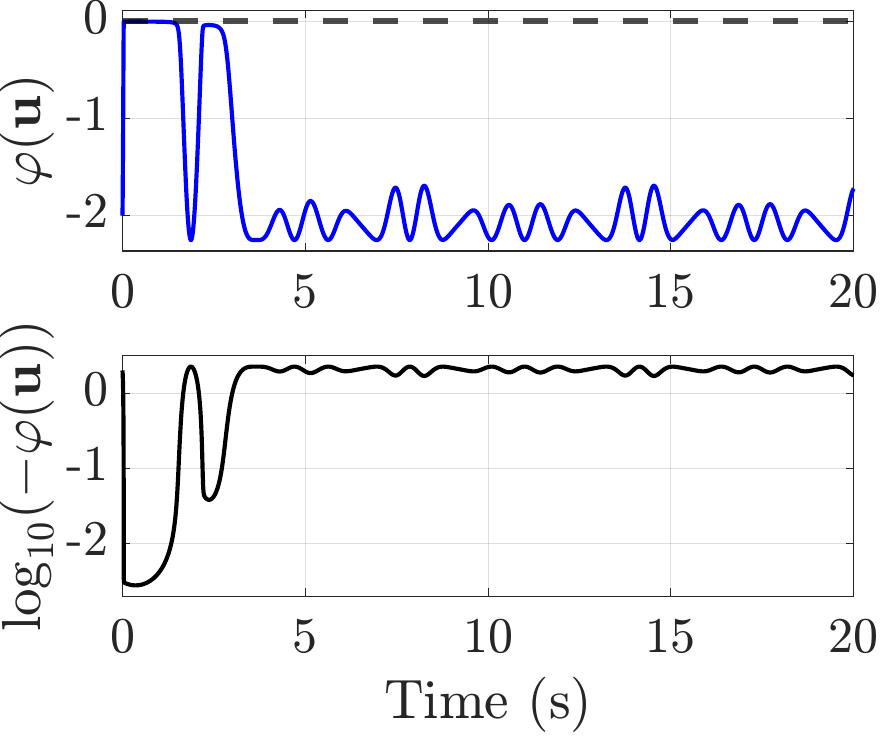}
        \caption{Joint capacity constraint \eqref{eq:joint_constraint}.}
        \label{fig:joint1}
    \end{subfigure}
\caption{Output tracking under the AJ-APIR framework}
    \label{fig:case1_results}
\end{figure*}
The corresponding results are presented in \Cref{fig:case1_results}. As shown in \Cref{fig:traj1}, the system's output $y(t)$ converges to the reference trajectory $y_d(t)$ after a brief transient of approximately \tb{$3 \rm{s}$}. During the initial transient phase, the system requires high control effort to nullify the tracking errors, resulting in a higher commanded input $\mathbf{u}_c(t)$, as evident from \Cref{fig:u1}. However, once the error is nullified, the control input demand decreases significantly. Note that despite this aggressive initial demand, the realized control inputs $u_1(t)$ and $u_2(t)$ strictly respect their individual asymmetric bounds at all times, as indicated by the dotted lines.
The input phase portrait in \Cref{fig:phase1} confirms that the control trajectory remains entirely confined within the joint admissible set $\mathcal{U}_{\mathrm{joint}}$. The joint constraint evaluation ${\varphi}(\mathbf{u})$, shown in \Cref{fig:joint1}, closely approaches the critical boundary ${\varphi}(\mathbf{u}) = 0$ during peak demand but remains strictly non-positive throughout, as further verified by the $\log_{10}(-\varphi(\mathbf{u}(t)))$ plot, confirming that no joint constraint violation occurs.

We further demonstrate the proposed framework on the three-dimensional
path-following guidance problem considered in \cite{doi:10.2514/1.G008792}. The path-following problem is formulated as the pursuit of a virtual target moving along the desired path, with the relative range, elevation lead angle, and azimuth lead angle regulated to zero through the vehicle's linear speed and pitch and yaw angular rates. In \cite{doi:10.2514/1.G008792}, the guidance inputs are constrained independently using isotropic realization, with each input restricted to its prescribed interval. However, such channel-wise constraints do not capture the common maneuvering authority shared by the pitch and yaw rate channels. To account for this coupling, we impose a joint constraint on the angular-rate commands and implement the resulting guidance inputs using AJ-APIR. The accompanying video  \href{https://youtu.be/21RK1Gd4fkY}{https://youtu.be/21RK1Gd4fkY} demonstrates the closed-loop response, where the vehicle converges to and follows the desired path while the realized angular rates satisfy both their individual bounds and the joint capacity constraint throughout the maneuver.
\section{Conclusions}\label{sec:conclusion}
This work developed an anisotropic admissibility-preserving control framework for multi-input strict-feedback nonlinear systems subject to a joint capacity constraint. The AJ-APIR uses the geometry of the joint constraint boundary to separate normal and tangential control directions. The normal component is attenuated as the boundary is approached, while the tangential component remains active, which prevents violation of joint constraints without unnecessarily restricting admissible control redistribution. The proposed AJ-APIR was then integrated with a recursive backstepping design, which demonstrates forward invariance of the joint admissible set and derives an explicit compact invariant set on which the realization channel remains uniformly nondegenerate. Under the proposed AJ-APC compatibility condition, the closed-loop system achieves exponential convergence of the tracking errors, while all closed-loop signals and the commanded input remain uniformly bounded.  The numerical simulations show that satisfying individual actuator bounds does not ensure joint admissibility, whereas AJ-APIR preserves the shared capacity constraint during tracking.
\bibliographystyle{IEEEtran}
\bibliography{Ref}

@article{kumar2026admissibility, title={Admissibility-Preserving Control for Strict-Feedback Nonlinear Systems with Asymmetric Actuator Constraints}, author={Kumar, Saurabh and Kumar, Shashi Ranjan and Sinha, Abhinav}, journal={arXiv preprint arXiv:2608.15375}, year={2026}}

@article{johansen2013control, title={Control allocation—A survey}, author={Johansen, Tor A and Fossen, Thor I}, journal={Automatica}, volume={49}, number={5}, pages={1087--1103}, year={2013}, publisher={Elsevier}}

@article{teel1999anti, title={Anti-windup for exponentially unstable linear systems}, author={Teel, Andrew R}, journal={International Journal of Robust and Nonlinear Control: IFAC-Affiliated Journal}, volume={9}, number={10}, pages={701--716}, year={1999}, publisher={Wiley Online Library}}

@article{zaccarian2002common, title={A common framework for anti-windup, bumpless transfer and reliable designs}, author={Zaccarian, Luca and Teel, Andrew R}, journal={Automatica}, volume={38}, number={10}, pages={1735--1744}, year={2002}, publisher={Elsevier}}

@article{grimm2003antiwindup, title={Antiwindup for stable linear systems with input saturation: an LMI-based synthesis}, author={Grimm, Gene and Hatfield, Jay and Postlethwaite, Ian and Teel, Andrew R and Turner, Matthew C and Zaccarian, Luca}, journal={IEEE Transactions on Automatic Control}, volume={48}, number={9}, pages={1509--1525}, year={2003}, publisher={IEEE}}

@article{wen2011robust, title={Robust adaptive control of uncertain nonlinear systems in the presence of input saturation and external disturbance}, author={Wen, Changyun and Zhou, Jing and Liu, Zhitao and Su, Hongye}, journal={IEEE Transactions on Automatic Control}, volume={56}, number={7}, pages={1672--1678}, year={2011}, publisher={IEEE}}

@article{tee2009barrier, title={Barrier Lyapunov functions for the control of output-constrained nonlinear systems}, author={Tee, Keng Peng and Ge, Shuzhi Sam and Tay, Eng Hock}, journal={Automatica}, volume={45}, number={4}, pages={918--927}, year={2009}, publisher={Elsevier}}

@inproceedings{ngo2005integrator, title={Integrator backstepping using barrier functions for systems with multiple state constraints}, author={Ngo, Khoi B and Mahony, Robert and Jiang, Zhong-Ping}, booktitle={Proceedings of the 44th IEEE Conference on Decision and Control}, pages={8306--8312}, year={2005}, organization={IEEE}}

@article{liu2016barrier, title={Barrier Lyapunov functions-based adaptive control for a class of nonlinear pure-feedback systems with full state constraints}, author={Liu, Yan-Jun and Tong, Shaocheng}, journal={Automatica}, volume={64}, pages={70--75}, year={2016}, publisher={Elsevier}}

@article{ames2016control, title={Control barrier function based quadratic programs for safety critical systems}, author={Ames, Aaron D and Xu, Xiangru and Grizzle, Jessy W and Tabuada, Paulo}, journal={IEEE Transactions on Automatic Control}, volume={62}, number={8}, pages={3861--3876}, year={2016}, publisher={IEEE}}

@inproceedings{ames2019control, title={Control barrier functions: Theory and applications}, author={Ames, Aaron D and Coogan, Samuel and Egerstedt, Magnus and Notomista, Gennaro and Sreenath, Koushil and Tabuada, Paulo}, booktitle={2019 18th European Control Conference (ECC)}, pages={3420--3431}, year={2019}, organization={IEEE}}

@inproceedings{rauscher2016constrained, title={Constrained robot control using control barrier functions}, author={Rauscher, Manuel and Kimmel, Melanie and Hirche, Sandra}, booktitle={2016 IEEE/RSJ International Conference on Intelligent Robots and Systems (IROS)}, pages={279--285}, year={2016}, organization={IEEE}}

@article{durham1993constrained, title={Constrained control allocation}, author={Durham, Wayne C}, journal={Journal of Guidance, Control, and Dynamics}, volume={16}, number={4}, pages={717--725}, year={1993}}

@article{bodson2002evaluation, title={Evaluation of optimization methods for control allocation}, author={Bodson, Marc}, journal={Journal of Guidance, Control, and Dynamics}, volume={25}, number={4}, pages={703--711}, year={2002}}

@article{sinha2026networked, title={Networked Admissibility-Preserving Control for Directed Safe Coordination}, author={Sinha, Abhinav and Gopikannan, Lohitvel and Kumar, Shashi Ranjan}, journal={arXiv preprint arXiv:2609.09384}, year={2026}}

@article{petersen2006constrained, title={Constrained quadratic programming techniques for control allocation}, author={Petersen, John AM and Bodson, Marc}, journal={IEEE Transactions on Control Systems Technology}, volume={14}, number={1}, pages={91--98}, year={2006}, publisher={IEEE}}

@article{zaccarian2009dynamic, title={Dynamic allocation for input redundant control systems}, author={Zaccarian, Luca}, journal={Automatica}, volume={45}, number={6}, pages={1431--1438}, year={2009}, publisher={Elsevier}}

@article{harkegaard2005resolving, title={Resolving actuator redundancy—optimal control vs. control allocation}, author={H{\"a}rkeg{\aa}rd, Ola and Glad, S Torkel}, journal={Automatica}, volume={41}, number={1}, pages={137--144}, year={2005}, publisher={Elsevier}}

@article{garone2017reference, title={Reference and command governors for systems with constraints: A survey on theory and applications}, author={Garone, Emanuele and Di Cairano, Stefano and Kolmanovsky, Ilya}, journal={Automatica}, volume={75}, pages={306--328}, year={2017}, publisher={Elsevier}}

@book{rawlings2020model, title={Model predictive control: theory, computation, and design}, author={Rawlings, James B and Mayne, David Q and Diehl, Moritz M}, publisher={Nob Hill Publishing}, year={2020}}

@book{krstic1995nonlinear, author = {Krsti\'{c}, Miroslav and Kanellakopoulos, Ioannis and Kokotovi\'{c}, Petar V.}, title = {Nonlinear and Adaptive Control Design}, publisher = {Wiley}, address = {New York}, year = {1995}}

@article{chen2011adaptive, author = {Chen, Mou and Ge, Shuzhi Sam and Ren, Beibei}, title = {Adaptive tracking control of uncertain {MIMO} nonlinear systems with input constraints}, journal = {Automatica}, volume = {47}, number = {3}, pages = {452--465}, year = {2011}}

@incollection{boyd2004convex_chapter,author= {Boyd, Stephen and Vandenberghe, Lieven}, title= {Convex Sets}, booktitle = {Convex Optimization}, chapter = {2}, section = {2.5.2}, pages = {50--51},publisher    = {Cambridge University Press}, year = {2004}, address = {Cambridge, UK}}

@article{kumar2024robust, title={Robust path-following guidance for an autonomous vehicle in the presence of wind}, author={Kumar, Saurabh and Sinha, Abhinav and Kumar, Shashi Ranjan}, journal={Aerospace Science and Technology}, volume={150}, pages={109225}, year={2024}, publisher={Elsevier}}

@article{kumar2024three, title={Three-dimensional path-following nonlinear guidance for unmanned aerial vehicles}, author={Kumar, Saurabh and Kumar, Shashi Ranjan and Sinha, Abhinav}, journal={Journal of Guidance, Control, and Dynamics}, volume={47}, number={6}, pages={1231--1240}, year={2024}, publisher={American Institute of Aeronautics and Astronautics}}

@article{doi:10.2514/1.G008792, author = {Kumar, Saurabh and Kumar, Shashi Ranjan and Sinha, Abhinav}, title = {Three-Dimensional Nonlinear Path-Following Guidance with Bounded Input Constraints}, journal = {Journal of Guidance, Control, and Dynamics}, volume = {49}, number = {1}, pages = {198-215}, year = {2026}, doi = {10.2514/1.G008792}}

\end{document}